\documentclass[sigconf, nonacm]{acmart}
\usepackage{xspace}
\usepackage{tabularx}
\newcolumntype{C}{>{\centering\arraybackslash}X} 
\usepackage{tcolorbox}
\tcbuselibrary{breakable} 
\newcommand{\sys}{\textsc{Scout}\xspace}

\usepackage{etoolbox}
\newtoggle{longversion}
\toggletrue{longversion}
\newcommand{\longonly}[1]{\iftoggle{longversion}{#1}{}}
\newcommand{\shortonly}[1]{\iftoggle{longversion}{}{#1}}
\makeatletter
\providecommand\fps@figure{tbp}
\providecommand\fps@dblfigure{tbp}
\iftoggle{longversion}%
  {\renewcommand\fps@figure{!t}\renewcommand\fps@dblfigure{!t}}%
  {\renewcommand\fps@figure{t}\renewcommand\fps@dblfigure{t}}%
\makeatother

\theoremstyle{acmdefinition}
\newtheorem{problem}{Problem}

\theoremstyle{acmplain}
\newtheorem{property}{Property}

\definecolor{promptblue}{RGB}{36,94,166}
\newcommand{\ph}[1]{\textcolor{promptblue}{[#1]}}

\usepackage{pgfplots}
\pgfplotsset{compat=1.17}
\usetikzlibrary{arrows.meta}

\definecolor{ansblue}{RGB}{0,90,200}
\definecolor{fig1green}{RGB}{100,180,50} 
\newcommand{\exhl}[1]{{\color{ansblue}\itshape #1}}   
\newcommand{\exarrow}{\par\vspace{-2.2mm}\centerline{\tikz{\draw[-{Latex[length=1.6mm]},line width=0.9pt] (0,0)--(0,-0.20);}}\vspace{-2.2mm}\par}

\newcommand{\exq}[2]{\begin{tcolorbox}[exquery]\textbf{#1:}~\textit{#2}\end{tcolorbox}} 
\tcbset{
  exquery/.style={colback=white, colframe=black!45, boxrule=0.4pt, arc=1.5pt,
    boxsep=0.5pt, left=3pt,right=3pt,top=1pt,bottom=1pt, halign=left, fontupper=\scriptsize},
  exsrcok/.style={colback=fig1green!45, colframe=fig1green!80!black, boxrule=0.8pt, arc=1.5pt,
    boxsep=0.5pt, left=3pt,right=3pt,top=1pt,bottom=1pt, fontupper=\scriptsize, valign=top},
  exsrcno/.style={colback=red!7, colframe=red!70!black, boxrule=0.6pt, arc=1.5pt,
    boxsep=0.5pt, left=3pt,right=3pt,top=1pt,bottom=1pt, fontupper=\scriptsize, valign=top},
  exsrc/.style={colback=white, colframe=black!55, boxrule=0.6pt, arc=1.5pt,
    boxsep=0.5pt, left=3pt,right=3pt,top=2pt,bottom=2pt, fontupper=\scriptsize, valign=top},
  excode/.style={colback=white, colframe=black!40, boxrule=0.4pt, arc=1pt,
    boxsep=0.5pt, left=3pt,right=3pt,top=1.5pt,bottom=1.5pt, fontupper=\scriptsize, valign=top,
    before skip=2pt, after skip=2pt},
}

\usepackage[ruled]{algorithm2e}

\usepackage{listings}
\lstdefinestyle{rulecode}{
  language=Python,
  basicstyle=\footnotesize\ttfamily,
  commentstyle=\color{black!55}\itshape,
  keywordstyle=\bfseries,
  stringstyle=\ttfamily,
  showstringspaces=false,
  breaklines=true,
  keepspaces=true,
  columns=fullflexible,
  aboveskip=2pt,belowskip=2pt,
}

\newcommand\vldbdoi{XX.XX/XXX.XX}
\newcommand\vldbpages{XXX-XXX}
\newcommand\vldbvolume{XX}
\newcommand\vldbissue{X}
\newcommand\vldbyear{2027}
\newcommand\vldbauthors{\authors}
\newcommand\vldbtitle{\shorttitle}
\newcommand\vldbavailabilityurl{https://github.com/yiminl18/SCOUT}
\newcommand\vldbpagestyle{plain}

\setcopyright{none}
\renewcommand\footnotetextcopyrightpermission[1]{}

\begin{document}
\title{\sys: Scalable Document Extraction via Data Similarity}

\author{Yiming Lin$^{*}$, Chiyu Hao$^{\dagger}$, Shreya Shankar$^{*}$, Aditya G. Parameswaran$^{*}$}
\affiliation{%
  \institution{$^{*}$UC Berkeley, $^{\dagger}$Shanghai Jiao Tong University \\ \{yiminglin, shreyashankar, adityagp\}@berkeley.edu, haochiyu84@gmail.com}
  \country{}
}

\begin{abstract}
Extracting values from large document collections powers data analysis across a variety of domains. Frontier LLMs extract such values accurately, but processing an entire collection using one is prohibitively costly. Yet this cost is largely avoidable: real-world collections exhibit rich similarity, so for the same query over a set of similar documents, the answer tends to recur in similar locations; 
an LLM need only read that small span, not the whole document. Prior methods that exploit this similarity fall short: they either assume a rigid document structure, or assume the answer is a set of substrings of the input and use an LLM-generated program to return it directly. Even a frontier agent fails to generate effective programs to directly locate the answer's span, as the search space is large and programs learned from a small sample tend to be  overfitted. We present \sys, a tool that generates accurate and cost-effective programs (that we call rules) to extract data at scale. From a few sampled documents, \sys generates a broad rule set and refines it by selecting a {\em pareto-optimal} subset with low cost without sacrificing accuracy. We prove this rule refinement to be NP-hard and present a greedy solution with a provable approximation guarantee. \sys can handle collections that are only partly similar, where similarity holds within clusters of documents. For such a setting, a sampling strategy, using no LLM, extracts samples from each cluster of similar documents; and a cascade strategy selects a subset of refined rules, falling back to the unrefined rule set when the selected rules don't contain the answer. Experiments on six real-world datasets show that \sys matches the accuracy of the strongest baseline, a frontier LLM agent that reads each full document, while being $61\times$ to over $1000\times$ cheaper on a collection of $1{,}000$ documents, and is $61\%$ more accurate than the strongest prior program-based approach.

\end{abstract}

\maketitle

\vspace{-2mm}
\pagestyle{\vldbpagestyle}
\begingroup\small\noindent\raggedright\textbf{PVLDB Reference Format:}\\
\vldbauthors. \vldbtitle. PVLDB, \vldbvolume(\vldbissue): \vldbpages, \vldbyear.\\
\href{https://doi.org/\vldbdoi}{doi:\vldbdoi}
\endgroup
\begingroup
\renewcommand\thefootnote{}\footnote{\noindent
This work is licensed under the Creative Commons BY-NC-ND 4.0 International License. Visit \url{https://creativecommons.org/licenses/by-nc-nd/4.0/} to view a copy of this license. For any use beyond those covered by this license, obtain permission by emailing \href{mailto:info@vldb.org}{info@vldb.org}. Copyright is held by the owner/author(s). Publication rights licensed to the VLDB Endowment. \\
\raggedright Proceedings of the VLDB Endowment, Vol. \vldbvolume, No. \vldbissue\ %
ISSN 2150-8097. \\
\href{https://doi.org/\vldbdoi}{doi:\vldbdoi} \\
}\addtocounter{footnote}{-1}\endgroup

\vspace{.1cm}
\begingroup\small\noindent\raggedright\textbf{PVLDB Artifact Availability:}\\
The source code, data, and/or other artifacts have been made available at \url{\vldbavailabilityurl}.
\endgroup

\section{Introduction}
\label{sec:introduction}


\begin{figure}[t]
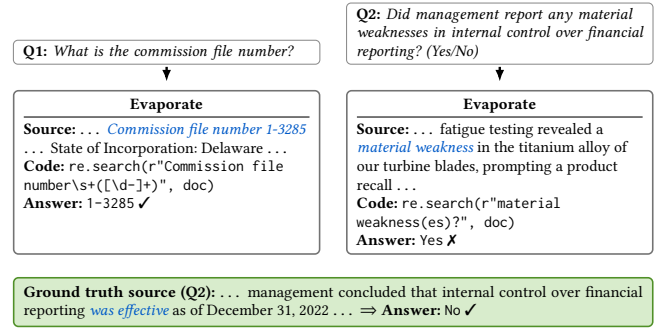

\centering
\begin{minipage}[t]{0.48\columnwidth}
\exq{Q1}{What is the commission file number?}
\exarrow
\begin{tcolorbox}[exsrc, equal height group=exrow1b]
\centering\textbf{Evaporate}\par\vspace{1.5pt}\hrule\vspace{2pt}\raggedright
\textbf{Source:} $\dots$ \exhl{Commission file number}~~\exhl{1-3285} $\dots$ State of Incorporation: Delaware $\dots$\\
\textbf{Code:} \texttt{re.search(r"Commission file number\textbackslash{}s+([\textbackslash{}d-]+)", doc)}\\
\textbf{Answer:} \texttt{1-3285}~\ding{51}
\end{tcolorbox}
\end{minipage}\hfill
\begin{minipage}[t]{0.48\columnwidth}
\exq{Q2}{Did management report any material weaknesses in internal control over financial reporting? (Yes/No)}
\exarrow
\begin{tcolorbox}[exsrc, equal height group=exrow1b]
\centering\textbf{Evaporate}\par\vspace{1.5pt}\hrule\vspace{2pt}\raggedright
\textbf{Source:} $\dots$ fatigue testing revealed a \exhl{material weakness} in the titanium alloy of our turbine blades, prompting a product recall $\dots$\\
\textbf{Code:} \texttt{re.search(r"material weakness(es)?", doc)}\\
\textbf{Answer:} \texttt{Yes}~\ding{55}
\end{tcolorbox}
\end{minipage}
\par\vspace{1mm}
\begin{tcolorbox}[exsrc, width=\columnwidth, colback=fig1green!25, colframe=fig1green!80!black]
\textbf{Ground truth source (Q2):} $\dots$ management concluded that internal control over financial reporting \exhl{was effective} as of December 31, 2022 $\dots$~~$\Rightarrow$~~\textbf{Answer:} \texttt{No}~\ding{51}
\end{tcolorbox}
\vspace{-4mm}
\caption{\small Extraction queries over a 10-K. Each panel shows the query, its
source span (answer-relevant parts in \exhl{blue italics}), and the program Evaporate generates
with the answer it returns. }
\label{fig:examples}
\vspace{-2mm}
\end{figure}

Data analytics in many domains centers on processing documents,
including SEC filings in finance~\cite{islam2023financebench}, federal court opinions in law~\cite{courtlistener}, and drug product labels in medicine~\cite{emaepar}. At the core of these analyses is extracting structured fields from the documents.  Table~\ref{tab:datasets} lists representative domains and the fields extracted from each. These tasks all hinge on extraction that is both accurate and scalable. 
Large language models (LLMs) have proven powerful for document extraction, but running them directly over large document collections is prohibitively costly. For example, extracting data specified in Table~\ref{tab:datasets} from 1M financial reports (e.g., from OfficeQA~\cite{opsahl2026officeqa}, each with  \textasciitilde{}110{,}000 tokens) with GPT-5.5 naively would cost  \textasciitilde{}\$275{,}000.
Pre-trained ML models, on the other hand, while inexpensive,  are narrow: each is tuned to one domain, does not generalize, and struggles with complex extraction. We need an approach that matches the accuracy of LLMs with the low cost to scale to large collections.


\begin{figure*}[t]
  \centering
  \includegraphics[width=\textwidth]{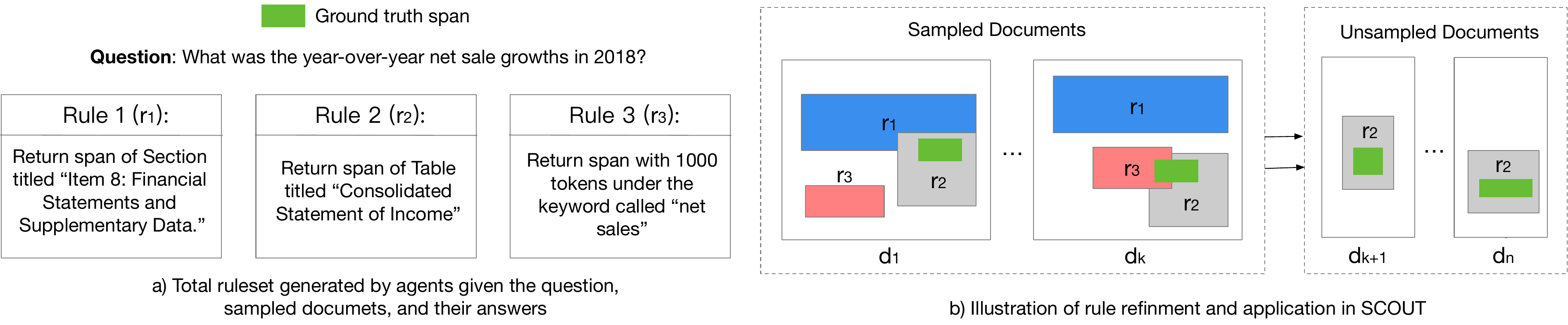}
  \vspace{-5mm}
  \caption{\small Motivating example to illustrate \sys's pipeline on financial datasets. }
  \vspace{-2mm}
  \label{fig:motivating_example}
\end{figure*}

Thankfully, real-world documents are rarely independent: {\em collections exhibit rich similarity that can be exploited for scale}, in two forms. First, documents may share similar {\em structure}, as in a template or layout, such as  articles from the same publisher or medical reports from the same hospital. Second, documents may share {\em domain semantics}, the same concepts and terminology, as research papers in the same field do. 
This similarity leads to one key insight: {\em for the same query over a set of similar documents, the answer tends to recur in similar locations}. Such similarity makes extraction scalable: once the answer's location has been learned, the LLM needs to only read that small span, rather than the complete document. 

\noindent{\bf Prior work leveraging data similarity falls short.} 
Prior work exploiting data similarity for scalable extraction falls into two categories. One assumes documents share a rigid structure, such as a tabular~\cite{lin2025visual} or hierarchical~\cite{zendb} template, or structured markup such as HTML tags of web pages~\cite{roadrunner,lixto,exalg,fivatech,webinvariants}, and uses it to locate and extract the target field. Such approaches fail when a collection follows no consistent template or markup. The other~\cite{evaporate} assumes each answer is a substring of the input, so that an LLM-generated program can extract it directly; e.g., in Figure~\ref{fig:examples}, a regex can recover the commission file number (Q1). Such approaches break down when the answer is not stated in the text and must instead be inferred from the context; for Q2, a program matching ``material weakness'' returns a false positive. Beyond the above, a large body of work studies extraction from each document individually, ignoring data similarity, and is thus either inaccurate at scale or not scalable; we defer detailed discussion to Section~\ref{sec:related}.

\noindent{\bf Frontier agent for data extraction program synthesis.} Deterministic programs are attractive for extraction because they are cheap and fast. One approach synthesizes a program to emit the answer directly, e.g.,~\cite{evaporate}, which is brittle. Instead, one can instruct a frontier agent (e.g., Codex~\cite{codex}) to  generate programs, or what we call {\em rules}, that return only the {\em subset of a document} where the answer lies, then let an LLM read that subset to produce the answer. Since the LLM, not the program, produces the final answer, this approach is  more flexible, at the cost of a little extra reasoning over the located span. For the net-sales query in Figure~\ref{fig:motivating_example}, an agent may generate rules $r_1$--$r_3$, each returning a candidate location: $r_1$ the section titled ``Item 8$\dots$'', $r_2$ the income-statement table, and $r_3$ the 1000 tokens around the keyword ``net sales''. The green {\em ground-truth span} is where the answer lies, and an oracle (e.g., a top-of-the-line LLM) can extract the answer from it. If a rule such as $r_2$ returns a small superset of this span, an oracle reading only that span, rather than the full document, yields correct answers at low cost, whenever such a location pattern recurs across the collection.

So, {\em can a frontier agent generate effective rules that locate answers accurately?} We run a GPT-5.4 agent on the six datasets in Table~\ref{tab:intro-motivation} ({\em baseline agent}), giving it the documents, queries, and tools that measure each program's {\em accuracy} (whether the located span lets the oracle recover the answer) and {\em cost} (the fraction of the document fed to the oracle), with targets of high accuracy and low cost (setup described in Section~\ref{sec:experiments}). The agent typically samples a few documents, obtains the oracle's answer on each, iteratively generates rules until the targets are met, applies them to the unsampled documents, and is scored against the ground truth. As Table~\ref{tab:intro-motivation} shows, its set of rules is still {\bf \em 16 points less accurate} than what we call a {\em golden baseline}, which runs the same agent over each full document; and though far cheaper than that baseline, its {\bf \em cost is over $3\times$ what is achievable}: \$$0.016$ per document, versus the \$$0.0047$ our approach achieves at the golden baseline's accuracy.

This gap has two systematic causes. First, selecting a cheap yet accurate subset from the large space of candidate rules is NP-hard (Section~\ref{sec:rule-refinement}), and the agent searches this space heuristically, with no guarantee of near-optimality. Second, to keep the generation cost low, the agent learns rules from a small sample, so they may {\em overfit} and fail to generalize to the rest of the collection; enlarging the sample raises the generation cost. 

\begin{table}[t]
\centering
\footnotesize
\setlength{\tabcolsep}{3pt}
\renewcommand{\arraystretch}{1.15}
\begin{tabular*}{\columnwidth}{@{\extracolsep{\fill}}lcccccc@{}}
\toprule
 & \multicolumn{2}{c}{\textbf{Golden baseline}} & \multicolumn{2}{c}{\textbf{Baseline agent}} & \multicolumn{2}{c}{\textbf{\sys}} \\
\cmidrule(lr){2-3}\cmidrule(lr){4-5}\cmidrule(lr){6-7}
\textbf{Dataset (\# of docs)} & \textbf{Accu} & \textbf{\$/doc} & \textbf{Accu} & \textbf{\$/doc} & \textbf{Accu} & \textbf{\$/doc} \\
\midrule
Court (294)        & 0.918 & 0.24 & 0.751 & 0.008 & 0.925 & 0.0004 \\
FinanceBench (100) & 0.986 & 0.22 & 0.764 & 0.020 & 0.943 & 0.0096 \\
NoPV (242)         & 0.922 & 0.28 & 0.781 & 0.011 & 0.929 & 0.0039 \\
OfficeQA (200)     & 0.830 & 0.74 & 0.666 & 0.034 & 0.831 & 0.0031 \\
Product (200)      & 0.902 & 0.25 & 0.802 & 0.013 & 0.886 & 0.0040 \\
Tropic (200)       & 0.813 & 0.21 & 0.671 & 0.011 & 0.873 & 0.0072 \\
\midrule
\textbf{Average} & 0.895 & 0.32 & 0.739 & 0.016 & 0.898 & 0.0047 \\
\bottomrule
\end{tabular*}
\caption{\small Accuracy and cost of agentic data extraction; \$/doc amortizes program-generation cost over the collection. }
\vspace{-10mm}
\label{tab:intro-motivation}
\end{table}

\begin{table*}[t]
  \centering
  \footnotesize
  \begin{tabular}{@{}l p{0.34\textwidth} p{0.44\textwidth}@{}}
    \toprule
    \textbf{Dataset} & \textbf{Description} & \textbf{Extracted data} \\
    \midrule
    FinanceBench~\cite{islam2023financebench} & SEC financial filings (10-K and 10-Q) from major publicly traded companies. & Registrant name, state of incorporation, reporting period, total revenue and net income, total assets, stock exchange and trading symbol. \\
    \midrule
    Court~\cite{courtlistener} & U.S. federal court appeal opinions. & Docket numbers, presiding and panel judges, argument and filing dates, final disposition, majority opinion author. \\
    \midrule
    NoPV~\cite{phmsanopv} & PHMSA Notices of Probable Violation for pipeline safety. & Operator name, CPF case number, PHMSA region, cited CFR sections, corrective-action deadlines, inspection dates. \\
    \midrule
    OfficeQA~\cite{opsahl2026officeqa} & U.S. Treasury Bulletins of periodic financial reports. & Quarter and year, GDP growth, unemployment rate, federal deficit, gross federal debt, total federal receipts, debt held by the public. \\
    \midrule
    Publications~\cite{dasigi2021dataset} & Academic and scientific research papers. & Title, authors, affiliations, venue and year, DOI, datasets referenced. \\
    \midrule
    Medical Records~\cite{mimiciv} & Clinical patient records such as admission notes and discharge summaries. & Patient MRN, age and sex, admission and discharge dates, diagnoses (ICD codes), medications, attending physician. \\
    \midrule
    NHC Tropical Cyclone Reports~\cite{nhctcr} & Post-storm reports from the U.S. National Hurricane Center. & Cyclone name, basin-year identifier, report date, lead author, minimum central pressure, peak winds, direct deaths, total damage. \\
    \midrule
    EMA EPAR Product Information~\cite{emaepar} & EU drug product labels (Summary of Product Characteristics) from the European Medicines Agency. & Product name, active substance, pharmaceutical form, first therapeutic indication, ATC code, half-life, marketing authorisation holder. \\
    \bottomrule
  \end{tabular}
  \caption{\small Datasets from distinct domains, with representative values extracted per document. }
  \vspace{-5mm}
  \label{tab:datasets}
\end{table*}



\noindent{\bf \sys: robust and effective rule generation for scalable extraction.}  
To use agents for program synthesis while overcoming their lack of guarantees, we propose \sys\footnote{\sys stands for \emph{\textbf{Sc}alable d\textbf{o}c\textbf{u}men\textbf{t} extraction via data similarity}}, which generates {\bf \em accurate and cost-effective extraction rules with provable guarantees}, and is {\bf \em broadly applicable}, without requiring rigid document structure or answers being substrings of the input.
Developing \sys involves two challenges. The first is to learn, from a few sampled documents, a rule set that {\em generalizes} to the rest of the collection with high accuracy and low cost. The second is that a collection is rarely uniform: answers may recur in similar locations only within clusters of documents, so rules learned from one cluster may miss the others. \sys addresses both through a series of techniques, organized as the pipeline in Figure~\ref{fig:pipeline}.

Given the extraction query $Q$, a set of sampled documents (our sampling approach  is described shortly), and their answers produced by an oracle LLM, \sys first prompts an agent to generate a set of candidate rules ({\em rule generation} in Section~\ref{sec:rule-generation}). This step prioritizes correctness over cost. 
We say a set of rules is {\em correct} on a document if the oracle can reproduce the answer from the union of the spans returned by its rules.   In this step, \sys instructs the agent to generate as many rules as possible, each returning a subset of each document by capturing a recurring answer pattern, so that their union is correct on as many documents as possible. This is both critical and achievable: we empirically show later (Table~\ref{tab:rulegen-results}) on six real-world datasets that the generated rules, when applied, match the accuracy of the most accurate but expensive baseline that runs the oracle over each complete document. 

Although the complete set of rules \sys generates is correct on most documents, it may contain both inefficient rules, which return too much text (e.g., $r_1$ in Figure~\ref{fig:motivating_example}), and inaccurate rules, which often miss the answer (e.g., $r_3$). Applying all of them is therefore unnecessarily expensive. \sys performs {\em rule refinement} (Section~\ref{sec:rule-refinement}) to select the lowest-cost subset that remains correct on the sampled documents. 
Here, we assume the oracle is {\em monotone}: if a span is sufficient for it to reproduce the answer, so is any superset of that span. Building on this assumption, our key insight is that once a rule subset is correct on a document, adding more rules cannot make it incorrect, since the union of returned spans still contains the answer. We show that rule refinement is NP-hard via a reduction from Set Cover, so \sys uses a greedy, provably near-optimal algorithm with an approximation factor of $O(\ln n)$, where $n$ is the number of sampled documents. At each step it adds the rule correct on the most not-yet-covered documents per unit of additional cost, favoring rules correct on many documents, which curbs overfitting. 

A subset selected on the sampled documents is still not enough, since a rule correct on the sample may fail on an unseen document. \sys therefore constructs a {\em cascade} of nested rule subsets (detailed in {\em rule application} in Section~\ref{sec:rule-application}): the first refined rule, the first two rules, and so on, ordered by their accuracy-to-cost ratio. For each subset in order, a cheap {\em proxy} LLM checks whether the returned spans contain enough information to answer $Q$. If so, the oracle answers $Q$ from those spans; otherwise, \sys proceeds to the next subset. If no refined subset passes the proxy, \sys falls back to the complete set of rules, which is correct on nearly all documents. 

%

Finally, when a collection is not uniformly similar, \sys uses a {\em sampling} strategy (Section~\ref{sec:sampling}). Such a collection splits into {\em clusters}, i.e., groups of documents whose answers share the same location pattern (e.g., in the same section) and are covered by one rule; the sample must include one document per cluster. This is difficult as identifying the clusters requires both the rules and the oracle, to verify which rules are correct on which documents, and neither is known a priori. \sys instead approximates by identifying the clusters with no LLM calls, computing the distribution of embedding similarities between the query and a document's chunks, which is used to further approximate the similarity between documents. It then draws a sample by repeatedly choosing the document least similar to those already selected, to cover distinct clusters.

\begin{figure}[t]
  \centering
  \includegraphics[width=\columnwidth]{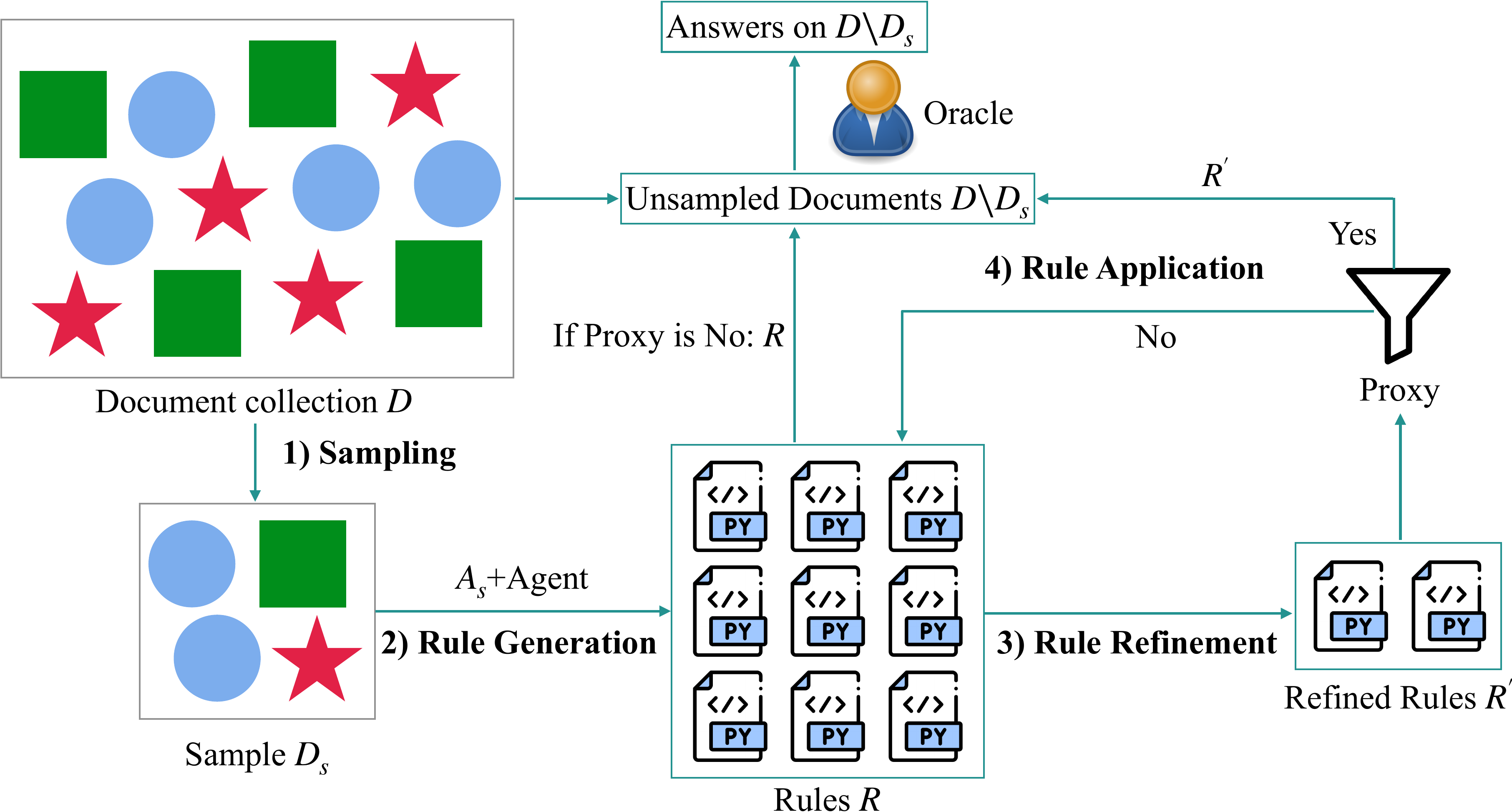}
  \vspace{-8mm}
  \caption{\small Pipeline of \sys. }
  \vspace{-8mm}
  \label{fig:pipeline} 
\end{figure}

We summarize our contributions as follows.

\begin{itemize}
    \item We formalize document similarity as {\em rules}, Python programs  synthesized by LLMs that retrieve portions of documents to be fed into an oracle LLM. Such rules serve as primitives for scalable document extraction. (Section~\ref{sec:problem})
    \item We formalize the selection of an accurate, cost-efficient rule set as a {\em pareto-optimal} optimization problem, prove it to be NP-hard, and provide a greedy algorithm with a provable approximation guarantee. (Section~\ref{sec:rule-refinement})
    \item We develop techniques that improve robustness for collections of varying similarity, including a sampling strategy and cascade rule application. (Section~\ref{sec:robustness})
    \item We show on six datasets that \sys matches the accuracy of a frontier agent that processes each complete document while being $61\times$ to over $1000\times$ cheaper at $1{,}000$ documents, and is $61\%$ more accurate than the strongest prior program-based approach. (Section~\ref{sec:experiments})
\end{itemize}




\section{Rule Definition and Overview of \sys} 
\label{sec:problem}

In this section, we formalize the concept of rules and then give an overview of \sys.

\subsection{Preliminaries}
\label{sec:problem-setup}

Consider a data extraction query $Q$ in natural language and a document collection $\mathcal{D} = \{d_1, d_2, \dots, d_n\}$. Let $O$ be an oracle that takes a document $d_j$ as input and produces the answer $A_j$ to $Q$, denoted $O(d_j,Q) = A_j$.
The oracle $O$ is a ground-truth labeler (e.g., a human) or a top-of-the-line LLM (e.g., GPT-5.5). 
Given $\mathcal{D}$ and $Q$, our goal is to produce $A_j$ for every $d_j \in \mathcal{D}$, matching the accuracy of the oracle but at minimal cost. For now, we focus on a single query $Q$ over the collection $\mathcal{D}$; we discuss multiple queries in Section~\ref{sec:experiments}. Here, both $|\mathcal{D}|$ and the size of an individual document can be large, making answering $Q$ over the entire collection prohibitively costly.  


We treat each document $d_j$ as a sequence of words in their reading order, $d_j = [w_{j1}, w_{j2}, \dots, w_{jm}]$, obtained by serializing its source file (e.g., a PDF or Word document) with off-the-shelf OCR tools, e.g., docling~\cite{docling}. An image is treated as a special word that points to its location in the document, i.e., its page and bounding box, so an LLM can analyze it, say, to extract data from it.


\subsection{The Notion of Rules}
\label{sec:problem-rules}

To support accurate extraction at low cost, we observe that answers to the same query often recur in similar patterns across similar subsets of a  document collection. 
We model answer locations using {\em rules}, on which we build the notion of data similarity.

Let $r_i$ be a rule, implemented as a Python program synthesized by an agent. Given a document $d_j$, it returns a subsequence $d_j^{i} \subseteq d_j$. For example, for the query ``What was the year-over-year net sales growth in 2018?'' in Figure~\ref{fig:motivating_example}(a), rule $r_2$ returns the ``Consolidated Statement of Income'' table of a 10-K document $d_j$ as its subsequence $d_j^{2}$. 
However, not all rules 
are {\em effective}. Some return a span that misses the ground-truth answer, such as rule $r_3$ on document $d_1$ in Figure~\ref{fig:motivating_example}(a). Others return a span that contains the answer but is unnecessarily large, so the oracle still incurs a high cost, such as rule $r_1$, which returns the entire ``Item~8'' section. 
We define a few properties below to quantify how  effective a rule is.

\noindent{\bf Accuracy}. Given the query $Q$, a rule $r_i$ is {\em correct} on document $d_j$ if running the oracle $O$ on the returned subsequence $d_{j}^{i}$ produces the answer $O(d_j,Q)$, i.e., $O(d_{j}^{i},Q) = O(d_j,Q)$. We denote the accuracy of $r_i$ on $d_j$ as $a(r_i,d_j) \in \{0,1\}$. 
Now, we define the accuracy of a set of rules, also called a {\em rule set}, $R = \{r_1,\dots,r_m\}$ on document $d_j$ as $a(R,d_j)$. $R$ is correct on $d_j$ if running the oracle on the union of its returned subsequences $\bigcup_{r_i\in R}d_{j}^{i}$ produces the same answer, i.e., $O(\bigcup_{r_i\in R}d_{j}^{i},Q) = O(d_j,Q)$. Finally, the accuracy of $R$ over a document set $D = \{d_1,\dots,d_l\}$ is $a(R,D) = \frac{\sum_{d_j\in D}a(R,d_j)}{|D|}$, the fraction of documents in $D$ on which $R$ is correct. 

A good rule should also be {\em general}: it should be correct on many documents rather than a few, so that it captures a recurring answer pattern rather than a coincidence in a handful of documents. We measure this generality by the {\em coverage} of a rule $r_i$ over a document set $D$, which is basically its accuracy taken individually over $D$: $a(r_i,D) = \frac{\sum_{d_j\in D}a(r_i,d_j)}{|D|}$, the fraction of documents in $D$ on which $r_i$ is correct.

\noindent{\bf Cost}. Beyond being correct and general, a rule should also be {\em cost-effective}: the subsequence it returns should be small enough to keep the cost of the oracle on that subsequence low. We define the cost of a rule $r_i$ on document $d_j$ as the number of tokens in the subsequence it returns, denoted $\mathit{c}(r_i,d_j) = |d_j^{i}|$. We further define the {\em cost ratio} of $r_i$ on $d_j$ as $\mathit{cr}(r_i,d_j) = \frac{|d_j^{i}|}{|d_j|}$, the fraction of $d_j$'s tokens that $r_i$ retrieves. Similarly, $cr(R,d_j) = \frac{|\bigcup_{r_i\in R}d_{j}^{i}|}{|d_j|}$. Finally, the average cost ratio of a set of rules $R$ over a document collection $D$ is $cr(R,D) = \frac{\sum_{d_j\in D}cr(R,d_j)}{|D|}$. 

Finally, we state a {\em monotonicity} property that formalizes an assumption about oracle behavior. Given $Q$, if a subsequence $d_j^i \subseteq d_j$ produces the oracle answer, i.e., $O(d_j^i, Q) = O(d_j,Q)$, then any superset $S$ of it with $d_j^i \subseteq S \subseteq d_j$ also produces $O(d_j,Q)$, since adding extra text to $d_j^i$ never misses the span that produces the answer. This property carries over to rules: if a rule set $R'$ is correct on $d_j$, then any superset $R \supseteq R'$ is also correct on $d_j$, because applying $R$ returns the union $\bigcup_{r_i\in R} d_j^i$ to the oracle, which only grows as rules are added. Under this assumption, a set of rules $R$ over a document set $D$ has accuracy at least that of any of its subsets over $D$, i.e., $a(R,D) \ge a(R',D)$ for any $R' \subseteq R$. We formally state this as the {\em accuracy monotonicity} property below.

\begin{property}[Accuracy monotonicity]
\label{prop:acc-monotone}
For any rule sets $R' \subseteq R$ and any document $d_j$, $a(R',d_j) = 1 \Rightarrow a(R,d_j) = 1$. Consequently, $a(R',D) \le a(R,D)$ for any document set $D$. 
\end{property}

Even a top-of-the-line LLM (e.g., GPT-5.5) is not a true oracle, so monotonicity may not always hold. Such failures are rare in practice: a recent study~\cite{blip} (published at VLDB 2026) reports that monotonicity holds on over $94\%$ of documents across real-world datasets, and we therefore adopt it as an assumption. When monotonicity fails on a document, only that document is affected, so the guarantees that follow degrade by at most the fraction of documents on which monotonicity fails.
Having formalized rules, we now give an overview of \sys.




\subsection{Overview of \sys}
\label{sec:problem-overview} 
As shown in Figure~\ref{fig:pipeline}, given a data extraction query $Q$ and a document collection $\mathcal{D}$, \sys is a tool that consists of four steps to generate effective programs to extract data from $\mathcal{D}$. Given $Q$, \sys first samples a subset $D_s \subseteq \mathcal{D}$ (Section~\ref{sec:sampling}, described later), then uses an agent to generate a set of rules from $D_s$ and $Q$, prioritizing accuracy (Section~\ref{sec:rule-generation}).
\sys then refines these rules by selecting a subset that preserves their accuracy while reducing cost (Section~\ref{sec:rule-refinement}). Selecting such a subset from the large space of candidate rules is NP-hard, so \sys employs a greedy approach with a provable guarantee.
Because the rules are selected on the sample, a rule correct on $D_s$ may still be incorrect on an unseen document. To reduce this overfitting and to handle collections where not all documents are similar, \sys uses two strategies (Section~\ref{sec:robustness}): (1)~a sampling strategy that draws a small but representative document subset (Section~\ref{sec:sampling}), and (2)~a rule-application strategy that applies the refined rules using a cascade with a proxy model and a fallback (Section~\ref{sec:rule-application}).
\section{Rule Generation}
\label{sec:rule-generation}

We now introduce how to generate an initial rule set by using an agent. 
For now, we assume we are given a query $Q$ and a sampled subset $D_s \subseteq \mathcal{D}$, and defer how to construct $D_s$ to Section~\ref{sec:sampling}. 
Rule generation prioritizes accuracy over cost: it aims for a broad rule set with high accuracy over $\mathcal{D}$, even if the cost of applying all its rules is high. This allows the larger set of generated rules to be used as a fallback, when a small set of refined rules does not contain the answer (as we will see later).

\subsection{Rule Generation}
\label{sec:rulegen-method}


\noindent{\bf Input and tools.} Rule generation takes as input a query $Q$, sampled documents $D_s$, and their oracle answers $\{A_j : d_j \in D_s\}$, obtained by running the oracle $O$ on each $d_j \in D_s$. Although a document is a sequence of words, we do not expose this word sequence to the agent. Instead, let $\widehat{d}_j$ be $d_j$ enriched by open-source OCR tools (e.g., Docling). $\widehat{d}_j$ is a list of consecutive {\em spans} in reading order, $\widehat{d}_j = \langle o_{j1}, o_{j2}, \dots, o_{jk}\rangle$, where each span $o_{ji}$ is a contiguous sequence of words annotated with a type label (e.g., \texttt{section\_header}, \texttt{text}, \texttt{table}, \texttt{list\_item}), a page number, typographic attributes (e.g., a bold flag and font size), and structural metadata (e.g., for tables, the row and column cell grid, and the caption/title). The spans partition the words of $d_j$, so $\widehat{d}_j$ is a labeled segmentation over the words of the document, and a rule reads both this span structure and the span text (e.g., via regex or keyword matching) to decide which sequence of words to return, down to spans or sub-spans. The output is a rule set $R = \{r_1, \dots, r_m\}$, each a Python function that maps a document to a subsequence of $d_j$ (Section~\ref{sec:problem-rules}); Figure~\ref{fig:rule-example} shows one rule for the net-sales query, which returns the \texttt{table} spans in the ``Consolidated Statement of Income'' whose text or cells mention ``net sales'', regardless of the page it appears on.

Beyond documents, the agent is equipped (powered by GPT-5.4, but any frontier LLM would work) with a set of {\em tools}: its accuracy and cost-ratio metrics (Section~\ref{sec:problem-rules}), exposed as tools so the agent can measure how good a rule (or a ruleset) is, as well as the default agentic tools (in this case from OpenAI) for inspecting documents, such as loading and filtering spans by page, or keyword, searching tables, or retrieving spans by embedding similarity. Embedding-based retrieval is one pattern a rule can express, making retrieval-augmented extraction possible via rules; rules may additionally capture structural and positional patterns, such as a section, a table, or a page, that semantic similarity alone may miss.

\begin{figure}[t]
\centering
\begin{tcolorbox}[colback=white,colframe=black!60,boxrule=0.5pt,arc=0pt,left=5pt,right=5pt,top=4pt,bottom=4pt]
\begin{lstlisting}[style=rulecode]
def rule_net_sales(doc):
    # "Consolidated Statement of Income" tables
    # that report net sales.
    out = []
    for span in doc["texts"]:
        if span["label"] != "table":      # tables only
            continue
        path = span["structure"]["path_text"].lower()
        text = span["text"].lower()
        # keep only the Consolidated Statement of Income
        if "consolidated statement of income" not in path:
            continue
        # match in the rendered text or in the cells
        if "net sales" in text \
           or "net sales" in join_cells(span):  # ...
            out.append(span)
    return out                  # (returns [] on any error)
\end{lstlisting}
\end{tcolorbox}
\vspace{-6mm}
\caption{\small Python script of $r_2$ in Figure~\ref{fig:motivating_example}. }
\label{fig:rule-example}
\shortonly{\vspace{-6mm}}
\longonly{\vspace{-6mm}}
\end{figure}

\noindent{\bf Prompt specification.} The agent is instructed as a {\em document rule engineer} to emit Python-based functions with format \texttt{def} \\ \texttt{rule\_name(doc)} to return subsets of documents that may contain the answer (the full prompt is in \shortonly{~\cite{scoutextended}}\longonly{Figure~\ref{fig:rulegen-prompt}}). The agent is instructed to {\em maximize the accuracy} of the rule set while keeping each individual rule cheap. To this end, it is instructed to  generate as many rules as possible, each capturing a distinct data pattern (such as the hint categories below); this diversity reduces overfitting and keeps accuracy $a(R,\mathcal{D})$ high across $\mathcal{D}$. 
Concretely, the prompt states two targets, with accuracy taking priority: (i)~the accuracy of the rules on the sample $D_s$ is at least $0.95$, and (ii)~each rule's cost ratio is below $0.1$. 
Given $Q$, $D_s$, and their oracle answers, $\{A_j: d_j \in D_s\}$, the prompt lists seven hint categories for locating oracle answers in the sample $D_s$, namely physical location (page), semantic location (section header or path), keyword proximity, data features (table cells), typography (bold, large, or all-caps), structural position (heading level and depth), and an open-ended ``any other'' category. The prompt also requires each rule to be self-contained and to return an empty list when nothing matches. When unsure, an agent is instructed to have a rule to return extra spans rather than risk dropping the answer. 

\longonly{
\begin{figure}[t]
\centering
\begin{tcolorbox}[colback=white,colframe=black!60,boxrule=0.5pt,arc=0pt,left=5pt,right=5pt,top=4pt,bottom=4pt]
\footnotesize\ttfamily
\textbf{\sys-RuleGen-Prompt:} You are a document rule engineer working in the project repository. You are given a question \ph{Question} and, for each sampled document, its enriched span representation \ph{Documents} together with the ground-truth answer \ph{Answers}. Write Python span-retrieval rules \texttt{def rule\_name(doc)} that return the spans that contain the answer from the given collection.

\textbf{Objective:} maximize the accuracy of the rule pool, and keep the cost ratio of each individual rule low. Generate as many rules as possible to capture distinct data patterns.

\textbf{Targets (accuracy first):} (1)~accuracy of the total rules at least $0.95$ on the sampled documents; (2)~each rule's cost ratio below $0.1$; accuracy always takes priority over cost.

\textbf{Tools:} \ph{Accuracy specification} and \ph{Cost ratio specification} are callable tools. You may also use all default agent tools, \ph{default tool specifications},  e.g., load a document, filter spans by label, page, or keyword, search tables, retrieve spans by embedding similarity, and evaluate arithmetic.

\textbf{Workflow (iterative):} study where the answer sits in each sampled document; add one or more new rules to the current pool and measure their accuracy and cost with the tools; never discard an existing rule. Repeat, extending the pool, until the soft targets are met.

\textbf{Hints:} locate the answer using (1)~physical location (page), (2)~semantic location (section header or path), (3)~keyword proximity, (4)~data features (table cells), (5)~typography (bold, large, or all-caps), (6)~structural position (heading level and depth), and (7)~any other recurring pattern.

\textbf{Requirements:} each rule must be self-contained, must never raise, and must return an empty list when nothing matches. Prefer returning a few extra spans over missing the answer.
\end{tcolorbox}
\longonly{\vspace{-5mm}}
\caption{\small A concise version of the rule-generation prompt.}
\label{fig:rulegen-prompt}
\longonly{\vspace{-7mm}}
\end{figure}
}

\noindent{\bf Iterative agentic rule generation.} \sys generates the rule set with an {\em agent} that tests and iterates on candidate rules. Starting from an empty set, the agent initially generates as many rules as possible; in later iterations, it is instructed to generate rules for the sampled documents on which the current rule set is not yet correct, checking their accuracy and cost with the tools above. Agent continues 
 until the targets are met.

\subsection{Empirical Assessment: Rule Quality}
\label{sec:rulegen-eval}
So far, we have presented how an agent can produce a rule set $R$ based on sample $D_s$; now we empirically evaluate how accurate $R$ is over the entire collection $\mathcal{D}$. 
We evaluate agent-generated rules on six datasets from distinct domains, each using a sample (generated by our  strategy to be described in Section~\ref{sec:sampling}) and a Codex agent powered by GPT-5.4 as the oracle. We defer the setup details to Section~\ref{sec:experiments}. For each query, $R$ has on average $19$ to $41$ rules across the datasets (column \emph{Avg \# rules}), each encoding one recurring pattern of where the answer resides. 
Table~\ref{tab:rulegen-results} reports the accuracy $a(R,\mathcal{D})$ (column {\em \sys}) and cost ratio $\mathit{cr}(R,\mathcal{D})$ (column {\em Cost ratio}) of $R$ over the collection $\mathcal{D}$. We also present a {\em golden baseline} that runs the oracle $O$ over each entire document $d_j$, the most accurate but most expensive strategy. Table~\ref{tab:rulegen-results} reveals two findings.

\begin{table}[t]
\centering
\scriptsize
\setlength{\tabcolsep}{2pt}
\renewcommand{\arraystretch}{1.25}
\begin{tabular*}{\columnwidth}{@{\extracolsep{\fill}}lccccc@{}}
\toprule
\textbf{\shortstack[l]{Dataset\\(\# Sampled\,/\,\# Total)}} & \textbf{\sys} & \textbf{\shortstack{Golden\\baseline}} & \textbf{\shortstack{Cost\\ratio}} & \textbf{\shortstack{Cost ratio\\(refined rules)}} & \textbf{\shortstack{Avg \#\\rules}} \\
\midrule
\textbf{FinanceBench} (20/100) & 0.943 & 0.986 & 0.088 & \textbf{0.0056} & $\sim$35 \\
\textbf{Court} (20/294)        & 0.925 & 0.918 & 0.074 & \textbf{0.0047} & $\sim$40 \\
\textbf{NoPV} (20/242)         & 0.929 & 0.922 & 0.392 & \textbf{0.0373} & $\sim$41 \\
\textbf{OfficeQA} (20/200)     & 0.831 & 0.830 & 0.208 & \textbf{0.0035} & $\sim$33 \\
\textbf{Product} (20/200)      & 0.886 & 0.902 & 0.096 & \textbf{0.0153} & $\sim$19 \\
\textbf{Tropic} (20/200)       & 0.873 & 0.813 & 0.211 & \textbf{0.1015} & $\sim$28 \\
\bottomrule
\end{tabular*}
\caption{\small Rule-generation quality.}
\label{tab:rulegen-results}
\shortonly{\vspace{-10mm}}
\longonly{\vspace{-10mm}}
\end{table}

{\bf\em First, the agent-generated rules have high recall.} 
We measure the accuracy $a(R, \mathcal{D})$ of the returned rule set $R$. 
$a(R, \mathcal{D})$ for \sys's rules, $0.898$, matches the golden baseline's $0.895$. Hence the rules, taken together, lose almost no accuracy relative to running the oracle $O$ on the entire document.

{\bf\em Second, there is substantial room to reduce cost.} Even without refinement, the rule set $R$ is efficient: its union reproduces the golden-baseline accuracy while reading only a fraction of each document, a cost ratio $cr$ of $0.07$ to $0.39$ versus $\mathit{cr}=1$ for the full document. As we show in Section~\ref{sec:exp-results}, a refined, selective subset of $R$ is far cheaper still: its cost ratio (column \emph{Cost ratio (refined rules)}) reaches only $0.004$ to $0.102$, a $2\times$ to $60\times$ reduction over $R$'s at comparable accuracy. Achieving this reduction is the goal of rule refinement, described next.

\section{Rule Refinement}
\label{sec:rule-refinement}

As we saw, the agent generates a set of rules $R$ with high recall, yet applying all of $R$ returns far too much of each document, leaving substantial room to reduce cost. We now present {\em rule refinement}, which selects a small subset of $R$ that preserves accuracy while reducing cost.

\subsection{Problem Definition}
\label{sec:refine-problem} 
Given the set of rules $R$ generated by the agent over the sample $D_s$, for each rule $r_i \in R$, and more generally any subset $R' \subseteq R$, we can evaluate its  accuracy $a(R', D_s)$ and cost ratio $\mathit{cr}(R',D_s)$ over $D_s$, as defined in Section~\ref{sec:problem-rules}. The cost ratio is free to compute, as it only requires counting the fraction of tokens retrieved by a rule; accuracy, in contrast, requires invoking the LLM and is therefore more costly to estimate. 

Intuitively, the subset of rules we retain should be {\em cost-efficient}, selecting as little of each document as possible, and {\em accurate}, reproducing the oracle's answer on every document. 
These goals conflict. 
We therefore seek subsets that best trade off cost against accuracy on the sample (the {\em pareto-optimal} ones). Note that generalization to unsampled documents is not part of this objective, but is handled separately by our algorithm, which favors general, high-coverage rules (Section~\ref{sec:refine-algorithm}). Now we define a {\em pareto-optimal set}: 

\begin{definition}[Pareto-optimal subset]
\label{def:pareto}
Given subsets $R', R'' \subseteq R$, $R'$ {\em dominates} $R''$ if $\mathit{cr}(R', D_s) \le \mathit{cr}(R'', D_s)$ and $a(R', D_s) \ge a(R'', D_s)$, with at least one inequality strict. $R'$ is {\em pareto-optimal} if no subset dominates it. $\mathcal{P}$ is the set of all pareto-optimal subsets.
\end{definition}

For example, in Figure~\ref{fig:pareto} the circles are pareto-optimal and the crosses are not, with each cross being dominated by a cheaper circle that is  at least as accurate. The two labeled rule sets $R'_1$ and $R'_2$ are both pareto-optimal, yet neither dominates the other: $R'_1$ is more accurate while $R'_2$ is cheaper. 

\begin{problem}[Rule Refinement]
\label{prob:refine}
Given the set of rules $R$, the sample $D_s$, and an accuracy tolerance $\alpha \ge 0$, find the pareto-optimal subset of least cost whose accuracy stays within $\alpha$ of $R$'s accuracy on $D_s$:
\[
R^* = \operatorname*{arg\,min}_{R' \in \mathcal{P}}\ \mathit{cr}(R', D_s) \quad \text{s.t.} \quad a(R', D_s) \ge a(R, D_s) - \alpha .
\]
\end{problem}

The tolerance $\alpha$ controls how much accuracy we trade for cost. We center it at $R$'s accuracy because, by accuracy monotonicity (Property~\ref{prop:acc-monotone}), no subset of $R$ can exceed $a(R, D_s)$. 
Taking $\alpha = 0$ then yields the cheapest subset of $R$ that matches $R$'s accuracy; a larger $\alpha$ allows for a bounded drop in accuracy for further savings. Such an $R^*$ is pareto-optimal by definition. For example, if the tolerance admits both labeled subsets in Figure~\ref{fig:pareto} (both sit above the $a(R, D_s) - \alpha$ line), an answer to Problem~\ref{prob:refine} would return $R'_2$.  


\begin{figure}
\centering
\includegraphics[width=0.8\columnwidth]{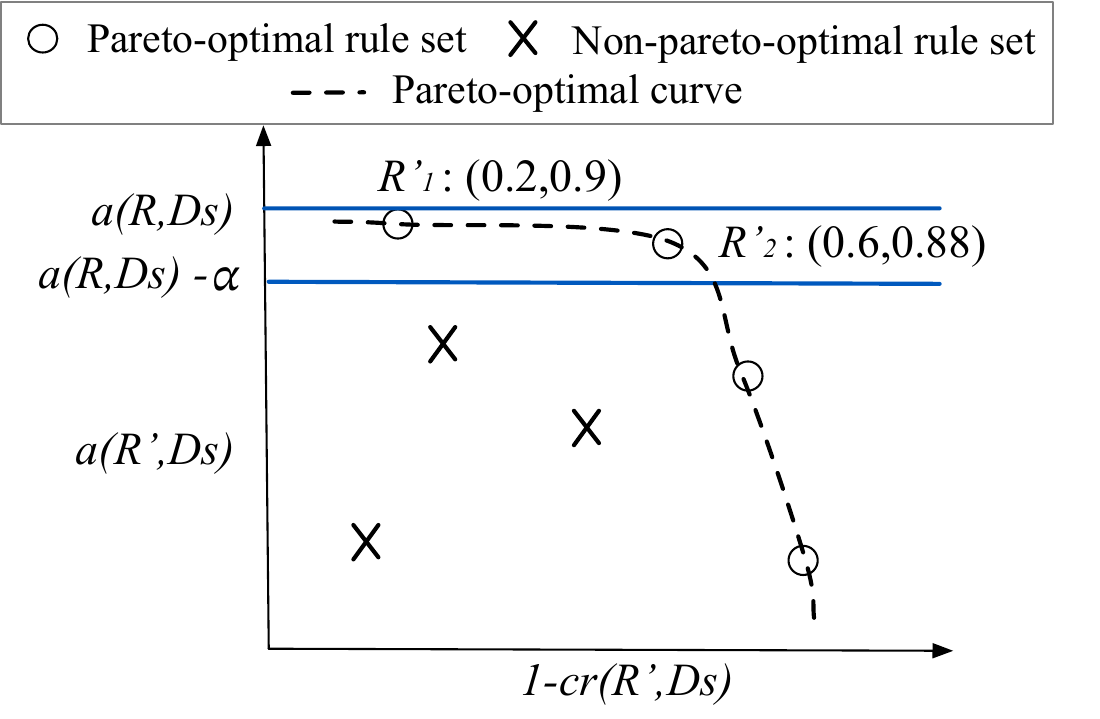}
\vspace{-5mm}
\caption{\small Rule refinement as a Pareto problem. }
\label{fig:pareto}
\vspace{-3mm}
\end{figure}

\noindent{\bf Complexity.} Problem~\ref{prob:refine} is intractable even in its simplest form. Consider $\alpha = 0$, where we ask for the cheapest subset that matches $R$'s accuracy exactly. We show this is already \textsf{NP}-hard by a simple reduction from \textsc{Set Cover}: each rule corresponds to a set and each document an element, matching $R$'s accuracy requires the chosen rules to cover every  document, and minimizing the cost ratio corresponds to minimizing the number of chosen rules. 
We state the result next \shortonly{ and defer the proof to the technical report~\cite{scoutextended}}.

\begin{theorem}
\label{thm:refine-hardness}
Rule Refinement (Problem~\ref{prob:refine}) is \textsf{NP}-hard, even for $\alpha = 0$.
\end{theorem}

\longonly{
\begin{proof}
Fix $\alpha = 0$ and consider the decision version of Problem~\ref{prob:refine}: given a budget $B$, decide whether some subset $R' \subseteq R$ has $\mathit{cr}(R', D_s) \le B$ and $a(R', D_s) \ge a(R, D_s)$. Since no subset can exceed the full pool's accuracy (Section~\ref{sec:refine-problem}), the accuracy constraint is equivalent to $a(R', D_s) = a(R, D_s)$. We reduce from \textsc{Set Cover}, which is \textsf{NP}-complete~\cite{karp1972reducibility}: given a universe $U = \{u_1, \dots, u_n\}$, a family $\mathcal{C} = \{C_1, \dots, C_m\}$ with $C_k \subseteq U$, and an integer $k$, decide whether $k$ of the sets cover $U$. We assume $\bigcup_k C_k = U$ (otherwise no cover exists and the instance is trivially negative).

\emph{Oracle semantics.} We instantiate the oracle as follows: every document carries a designated single-token {\em answer span}, and the oracle reproduces a document's answer exactly when the text it receives contains that span. This semantics is monotone (extra text never removes a correct answer), so it satisfies the assumption behind Property~\ref{prop:acc-monotone}; a rule set is correct on a document iff the union of its returned spans contains the answer span.

\emph{Construction.} In polynomial time we build a refinement instance. The sampled subcollection is $D_s = \{d_1, \dots, d_n\} \cup \{G_1, \dots, G_m\}$: one {\em element document} $d_j$ per element $u_j$, and one {\em ballast document} $G_k$ per set $C_k$. Every document has $T = 2mn + 1$ tokens, with its answer span at the last token. The pool has one rule per set, $R = \{r_1, \dots, r_m\}$, defined by what each rule returns:
\begin{itemize}
\item on element document $d_j$, rule $r_k$ returns $d_j$'s answer span if $u_j \in C_k$, and the empty span otherwise;
\item on ballast document $G_k$, rule $r_k$ returns the first $T-1$ tokens of $G_k$, missing its answer span, while every other rule returns the empty span.
\end{itemize}
Hence $\mathit{cr}(r_k, d_j) \in \{0, 1/T\}$, whereas $\mathit{cr}(r_k, G_k) = 1 - 1/T$ and $\mathit{cr}(r_{k'}, G_k) = 0$ for $k' \neq k$.

No rule's span contains a ballast document's answer span, so $a(R', G_k) = 0$ for every subset $R'$, including the full pool; the full pool answers exactly the $n$ element documents (each $d_j$ via some $C_k \ni u_j$), so $a(R, D_s) = \tfrac{n}{n+m}$. A subset $R'$ answers $d_j$ iff its union $\bigcup_{r_k \in R'} d_j^k$ contains $d_j$'s answer span, i.e., iff some $r_k \in R'$ has $u_j \in C_k$. Hence $a(R', D_s) = \tfrac{1}{n+m}\,\bigl|\{j : u_j \in \bigcup_{r_k \in R'} C_k\}\bigr|$, so
\[
a(R', D_s) = a(R, D_s) \iff \{C_k : r_k \in R'\} \text{ covers } U .
\] For the cost, the ballast document $G_k$ is private to rule $r_k$, so the union shares no cost across rules: each rule in $R'$ contributes $1 - 1/T$ through its ballast document, while the element documents add between $0$ and $n/T$ in total. Hence $\sum_{d \in D_s} \mathit{cr}(R', d) = |R'| + \delta$ with $-\tfrac{m}{T} \le \delta \le \tfrac{n}{T}$, so $|\delta| < \tfrac12$ since $T > 2mn$, and, since $|D_s| = n + m$,
\[
\mathit{cr}(R', D_s) = \frac{|R'| + \delta}{n + m} ,
\]
so $\mathit{cr}(R', D_s) \le \frac{k + \frac12}{n + m} \iff |R'| \le k$.

\emph{Equivalence.} Set the budget $B = \frac{k + \frac12}{n + m}$. If $\mathcal{C}$ has a cover of size at most $k$, the corresponding rules form a subset $R'$ with $a(R', D_s) = a(R, D_s)$ and $\mathit{cr}(R', D_s) \le B$. Conversely, any $R'$ with $a(R', D_s) = a(R, D_s)$ and $\mathit{cr}(R', D_s) \le B$ yields a cover $\{C_k : r_k \in R'\}$ of size $|R'| \le k$, since $|\delta| < \tfrac12$ and $|R'|$ is integral. The reduction is polynomial, so the decision problem is \textsf{NP}-hard. Hardness carries over to Problem~\ref{prob:refine} itself: at $\alpha = 0$ the feasible subsets are exactly those of maximum accuracy, so a minimum-cost feasible subset cannot be dominated and is pareto-optimal; solving Problem~\ref{prob:refine} thus yields the minimum feasible cost, which decides the budget question.
\end{proof}

The decision problem lies in \textsf{NP} whenever $a(R', D_s)$ and $\mathit{cr}(R', D_s)$ are polynomial-time computable, in which case it is \textsf{NP}-complete.
}


\setlength{\textfloatsep}{6pt}
\begin{algorithm}[bt]
\small
\DontPrintSemicolon
\LinesNumbered
\caption{Rule refinement (cost-effectiveness greedy)}
\label{alg:refine}
\KwIn{rule set $R$, sample $D_s$, tolerance $\alpha$}
\KwOut{refined subset $R'$}
\ForEach{$r \in R$}{
estimate $a(r,d_j)$ for all $d_j \in D_s$ with the cheap model; compute $\mathit{cr}(r,D_s)$ by token count\;\label{ln:score}
}
$R' \leftarrow \emptyset$;\quad $U \leftarrow \{d_j \in D_s : a(R,d_j)=1\}$\;\label{ln:init}
\While{$a(R',D_s) < a(R,D_s) - \alpha$}{\label{ln:stop}
$r^{*} \leftarrow \displaystyle\arg\max_{r \in R \setminus R'}\ \frac{|\{d_j \in U : a(r,d_j)=1\}|}{\mathit{cr}(r,D_s)}$\;\label{ln:pick}
$R' \leftarrow R' \cup \{r^{*}\}$;\quad $U \leftarrow U \setminus \{d_j : a(r^{*},d_j)=1\}$\;\label{ln:update}
}
\Return $R'$\;
\end{algorithm}

\subsection{Rule Refinement Algorithm}
\label{sec:refine-algorithm}

We now present Algorithm~\ref{alg:refine} to approximately solve Problem~\ref{prob:refine}. Given the rules $R$ from rule generation, \sys first labels, for every sampled document, which rules lead to a correct answer, and computes each rule's cost ratio (Line~\ref{ln:score}). Running the oracle (e.g., GPT-5.4) to estimate the accuracy for each rule would be expensive, so \sys uses a {\em cheap} proxy model (e.g., GPT-5.4-mini). The proxy is reliable because the text it reads per rule is small: even the union of all rules in $R$ has a cost ratio of only $0.07$ to $0.39$ (average $0.18$) across the datasets (the \emph{Cost ratio} column of Table~\ref{tab:rulegen-results}), and any single rule is smaller still. On such short inputs, the proxy matches the oracle's estimate of each rule's accuracy at a fraction of the cost. 

\sys then builds the subset greedily. The algorithm starts from an empty subset $R'$ and marks every {\em answerable} document, $U = \{d_j \in D_s : a(R, d_j) = 1\}$, as not yet covered (Line~\ref{ln:init}). In each step, \sys picks the rule with the best {\em cost-effectiveness}: the number of still-uncovered documents, divided by its cost ratio $\mathit{cr}(r, D_s)$, so a rule that answers many new documents cheaply is picked (Line~\ref{ln:pick}). \sys adds this rule $r^{*}$ to $R'$ and marks the documents it answers as covered  (Line~\ref{ln:update}). The algorithm repeats until $R'$'s accuracy $a(R', D_s)$ comes within $\alpha$ of $a(R, D_s)$ (Line~\ref{ln:stop}), then returns $R'$.


This greedy algorithm also alleviates overfitting as a side-effect. Our cost-effective rule notion (Line~\ref{ln:pick}) divides the number of documents in $U$ on which a rule is correct by that rule's cost ratio $\mathit{cr}(r,D_s)$. For two rules with the same cost ratio, \sys  selects the one correct on more documents of $U$, i.e., the rule with higher coverage $a(r,U)$ (Section~\ref{sec:problem-rules}). Such a rule is correct on many sampled documents because its answer pattern recurs across them, so it is more likely to be correct on the unsampled documents.

\noindent{\bf Approximation guarantee.} At $\alpha = 0$, Algorithm~\ref{alg:refine} is exactly the greedy algorithm for weighted \textsc{Set Cover}~\cite{vazirani2001approximation}: each answerable document is an element to cover, and each rule is a set weighted by its cost ratio $\mathit{cr}(r, D_s)$. We then present a provable guarantee below for this approximation algorithm. 


\begin{theorem}
\label{thm:refine-approx}
Let $n = |\{d_j \in D_s : a(R, d_j) = 1\}|$ be the number of answerable documents in $D_s$, and let $H_n = \sum_{i=1}^{n} \tfrac{1}{i} = O(\ln n)$. At $\alpha = 0$, the subset $R'$ returned by Algorithm~\ref{alg:refine} and the optimal $R^{*}$ of Problem~\ref{prob:refine} satisfy
\[
\mathit{cr}(R', D_s) \;\le\; H_n \sum_{r \in R^{*}} \mathit{cr}(r, D_s) .
\]
\end{theorem}

The theorem follows from Chv\'atal's analysis~\cite{chvatal1979greedy} of the greedy algorithm for weighted \textsc{Set Cover}. If the sample size $|D_s| = 20$ (assuming every sampled document is answerable), then $H_n \le 3.6$. 
\shortonly{ We defer the proof and the $\alpha > 0$ case to the technical report~\cite{scoutextended}.}

\longonly{

\begin{proof}
The first inequality is $\mathit{cr}(R', D_s) \le \sum_{r \in R'} \mathit{cr}(r, D_s)$, since on every document the union of the spans of $R'$ has at most as many tokens as those spans counted one rule at a time. For the second, Algorithm~\ref{alg:refine} at $\alpha = 0$ is Chv\'atal's greedy~\cite{chvatal1979greedy} on the weighted \textsc{Set Cover} instance whose universe is the $n$ answerable documents, with one set per rule $r$ holding the documents $r$ is correct on and weight $\mathit{cr}(r, D_s)$: Line~\ref{ln:pick} admits the rule of least weight per newly covered document, and the loop (Line~\ref{ln:stop}) ends once every answerable document is answered, since at $\alpha = 0$ this is exactly when the target $a(R, D_s)$ is met. Chv\'atal's theorem~\cite{chvatal1979greedy} bounds the total weight of the greedy's cover by $H_n$ times that of any feasible cover. At $\alpha = 0$ the constraint of Problem~\ref{prob:refine} forces $a(R^{*}, D_s) = a(R, D_s)$, so $R^{*}$ is correct on all $n$ answerable documents and is a feasible cover, giving $\sum_{r \in R'} \mathit{cr}(r, D_s) \le H_n \sum_{r \in R^{*}} \mathit{cr}(r, D_s)$.
\end{proof}
\vspace{-2mm}
For $\alpha > 0$, Algorithm~\ref{alg:refine} stops once it covers enough documents to stay within tolerance; this is greedy {\em partial cover}, and Slav\'ik's analysis~\cite{slavik1997partial} gives the analogous bound $H_{n'}$, where $n'$ is the number of documents the subset must cover. Finally, the labels $a(r, d_j)$ in Line~\ref{ln:score} come from a cheap proxy model, so the guarantee holds with respect to these labels, which are reliable on the short spans that rules return, as discussed above.
}

\section{Robustness of \sys}
\label{sec:robustness}

So far, rule refinement returns a rule set that is accurate and cost-efficient on the sample $D_s$. Multiple challenges remain, however, before these rules can be applied on  unsampled documents. One is that the rules learned on the sample may overfit, performing well on the sample but poorly on unseen documents. To reduce overfitting, we further introduce two strategies: (1) a sampling strategy that ensures that $D_s$ covers as many data patterns to locate answers as possible so that the learned rules are correct on $\mathcal{D}\setminus D_s$ (Section~\ref{sec:sampling}), and (2) a rule application strategy that falls back to the  rule set returned by the agent when the refined rules don't contain the correct answer (Section~\ref{sec:rule-application}). 
We first define {\em document similarity} by defining the concept of {\em document clusters} below, followed by presenting the above strategies. 


\noindent{\bf Document similarity.} Given a  query $Q$ and a set of documents $\mathcal{D}$, a document subset $C_j \subseteq \mathcal{D}$ is called a  {\em cluster} if there exists a rule $r_i$ that, when applied to $C_j$, incurs a small cost, i.e., $cr(r_i, C_j) < \delta$ ($\delta$ is a small constant fraction, say 0.1), and is correct on every document in $C_j$, i.e., $a(r_i,C_j)=1$. Such a rule $r_i$ on $C_j$ implies that documents in $C_j$ are {\em similar} since there exists a data pattern encoded by the rule $r_i$ (e.g., to return the span of a table whose caption matches a keyword) that can be used to extract data correctly with low cost from $C_j$. In this case, we also say $r_i$ {\em covers} the cluster $C_j$. Consider the example in Figure~\ref{fig:clusters}(a), where $\mathcal{D}$, represented by colored shapes, is split into three clusters, each covered by one of the rules $r_1$, $r_2$, and $r_3$, while the right cluster in Figure~\ref{fig:clusters}(b) is covered by $r_2 \cup r_3$ (returning the union of the spans of $r_2$ and $r_3$), where neither $r_2$ nor $r_3$ can cover this cluster individually.



\begin{figure}
\centering
\includegraphics[width=1\columnwidth]{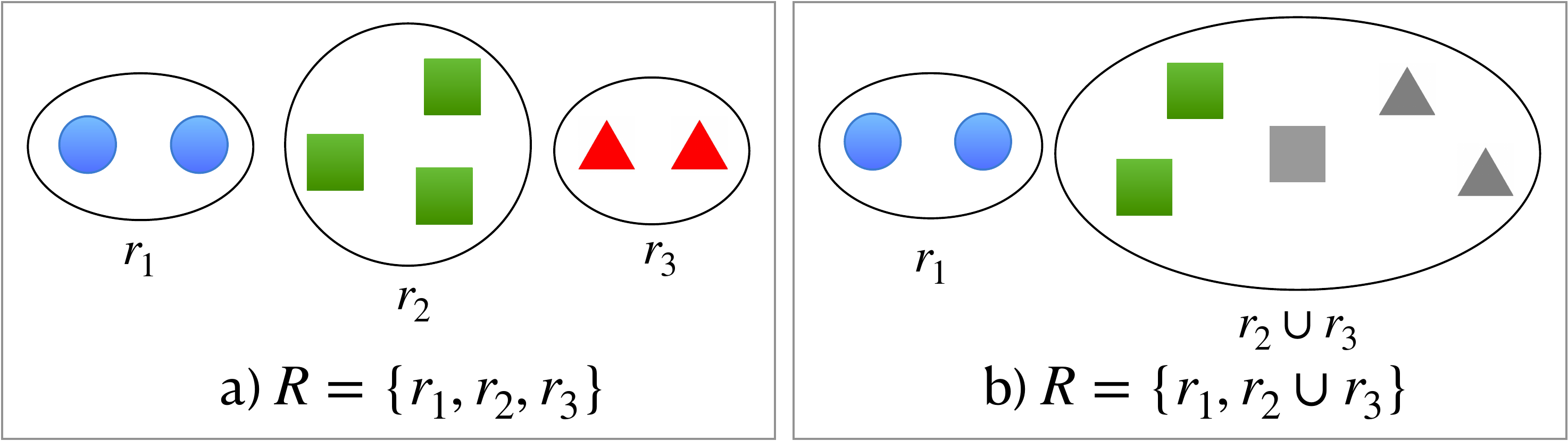}
\vspace{-7mm}
\caption{\small Documents cluster by where their answers to $Q$ reside, and each rule covers one cluster of similar documents.}
\label{fig:clusters}
\vspace{-2mm}
\end{figure}




\subsection{Sampling}
\label{sec:sampling}

Ideally, an effective sample draws documents from every cluster of $\mathcal{D}$, so that \sys sees all possible data patterns to locate answers and learns  rules that are less likely to overfit. 
Generating such a sample is non-trivial, as we do not have rules before sampling, and using the oracle to verify whether a rule is correct on a document subset is too expensive. To this end, we approximate a {\em document distance} that is used for sampling {\em without invoking the oracle}. 

\noindent{\bf Approximate Document Distance.} Our key insight is that a good distance metric should not measure how similar the contents of two documents $d_i$ and $d_j$ are, but whether their answers to the query $Q$ can be represented by using similar data patterns, e.g., in {\em similar} locations.  
Consider a document $d_i\in \mathcal{D}$ and a query $Q$. \sys splits $d_i$ into $m$ chunks and computes the embedding similarity $s_k$ between $Q$ and each chunk $k$. It stores these values in a vector $V_i = [s_1, \dots, s_m]$, ordered by the reading order of the chunks in $d_i$. $V_i$ approximates how similarity to $Q$ is distributed across chunks in $d_i$. Given two documents $d_i$ and $d_j$, we define the distance between their vectors $V_i$ and $V_j$ as the cosine distance
$\mathit{dis}_{i,j} = 1 - \frac{V_i \cdot V_j}{\lVert V_i \rVert\, \lVert V_j \rVert}$. 

Vector $V_i$ approximates the likelihood that document subsets (represented by chunks) contain the answer to $Q$ and the relative locations of the answers. Cosine distance $\mathit{dis}_{i,j}$ measures how similar the relative locations of answers to $Q$ are in documents $d_i$ and $d_j$. A small $\mathit{dis}_{i,j}$ indicates that $d_i$ and $d_j$ place their answers in similar relative positions, so a single rule that encodes a recurring data pattern (e.g., a section or a page) is likely to cover both. 
The number of chunks $m$ controls the granularity of the similarity curve: a larger $m$ gives a finer-grained curve. We empirically set $m = 50$, which yields a reliable estimate. While this distance metric is not exact, we show that it is effective.  When two documents in the same true cluster have a large $\mathit{dis}_{i,j}$ and thus are split apart, we usually end up picking too many documents rather than too few, which is permissible as we don't leave a cluster unrepresented. 





\setlength{\textfloatsep}{6pt}
\begin{algorithm}[t]
\small
\DontPrintSemicolon
\LinesNumbered
\caption{Farthest-point sampling}
\label{alg:fps}
\KwIn{collection $\mathcal{D}$, vectors $\{V_i\}$, distance $\mathit{dis}$, stop ratio $\rho = 0.5$}
\KwOut{sample $S$}
$d_1 \leftarrow$ document farthest from the centroid of $\mathcal{D}$\;\label{ln:fps-seed}
$S \leftarrow \{d_1\}$;\quad $g_{\mathrm{prev}} \leftarrow \infty$\;
\While{$S \neq \mathcal{D}$}{
$d^{*} \leftarrow \displaystyle\arg\max_{d_i \in \mathcal{D} \setminus S}\ \min_{d_j \in S} \mathit{dis}_{i,j}$\;\label{ln:fps-pick}
$g \leftarrow \displaystyle\min_{d_j \in S} \mathit{dis}_{d^{*},j}$\;\label{ln:fps-gap}
\If{$|S| \ge 2$ \textbf{and} $g < \rho \cdot g_{\mathrm{prev}}$}{\textbf{break}\;\label{ln:fps-stop}}
$S \leftarrow S \cup \{d^{*}\}$;\quad $g_{\mathrm{prev}} \leftarrow g$\;\label{ln:fps-update}
}
\Return $S$\;
\end{algorithm}


\noindent{\bf Farthest-point Sampling.} We now present Algorithm~\ref{alg:fps} to select a sample. \sys seeds the sample $S$ with the document farthest from the centroid of the collection $\mathcal{D}$ under $\mathit{dis}$ (Line~\ref{ln:fps-seed}). \sys then repeatedly takes the candidate whose distance to $S$ is the largest, where a document's distance to $S$ is its smallest distance to any document in $S$, and stores that distance as the {\em gap} $g = \max_{d_i \in \mathcal{D} \setminus S} \min_{d_j \in S} \mathit{dis}_{i,j}$ (Lines~\ref{ln:fps-pick}--\ref{ln:fps-gap}). \sys admits the candidate unless $g$ drops below $\rho$ times the previous gap, in which case it stops (Lines~\ref{ln:fps-stop}--\ref{ln:fps-update}). 

$g$ is non-increasing by construction, since adding a document to $S$ can only reduce the distance of every remaining document to $S$. While some cluster has no document in $S$, $g$ is at least the smallest distance between two documents in different clusters; once every cluster has a document in $S$, the next candidate is in a cluster that has already been sampled. Hence $g$ is at most the largest distance between two documents in the same cluster. 

We now analyze the returned sample $S$. Let $g^{*}$ be the value of the gap $g$ when Algorithm~\ref{alg:fps} stops at Line~\ref{ln:fps-stop}. 
Let $\mathcal{D}$ partition into clusters $\{C_1, \dots, C_k\}$ with $k \ge 2$, and measure all distances using $\mathit{dis}$. The {\em separation} of a document subset $U \subseteq \mathcal{D}$ is the smallest distance between a document in $U$ and a document outside $U$; for a cluster $C_j$, let $\delta_j = \min_{d_a \in C_j,\, d_b \notin C_j} \mathit{dis}_{a,b}$ be its separation.

%
%
%
%

\begin{theorem}[Cluster coverage]
\label{thm:fps}
The sample returned by Algorithm~\ref{alg:fps} contains at least one document from every cluster whose separation is greater than $g^{*}$.
\end{theorem}

\shortonly{ We give the proofs of Theorems~\ref{thm:fps} and~\ref{thm:fps-group} (to be described below) in~\cite{scoutextended}, with a sketch below.} Every document outside $S$ is within $g^{*}$ of some document in $S$ by  definition, and every two documents in $S$ are at distance greater than $2\,g^{*}$, since gaps never increase and $g^{*} < 0.5\,g_{\mathrm{prev}}$. If a cluster with separation exceeding $g^{*}$ had no document in $S$, all its documents would be farther than $g^{*}$ from $S$. 
The same argument applies to the union of two clusters, which yields a stronger guarantee.

\begin{theorem}[Group coverage]
\label{thm:fps-group}
For two clusters $C_i$ and $C_j$, if the separation of $C_i \cup C_j$ is greater than $g^{*}$, then the sample returned by Algorithm~\ref{alg:fps} contains at least one document from $C_i$ or from $C_j$.
\end{theorem}

Theorem~\ref{thm:fps} does not hold when a cluster, say $C_i$, has separation at most $g^{*}$. In this case, let $C_j$ be the cluster that is close to $C_i$, i.e., $\min_{d_a \in C_i,\, d_b \in C_j} \mathit{dis}_{a,b} \le g^{*}$. By Theorem~\ref{thm:fps-group}, if $C_i \cup C_j$ is far from the rest, i.e., the separation of $C_i \cup C_j$ is greater than $g^{*}$, then $C_i$ and $C_j$ cannot both be missed. This indicates that when Theorem~\ref{thm:fps} does not hold, at least one document in $C_j$ that is {\em similar}  to those in $C_i$ will be sampled, so that documents far from all others are not left unrepresented, ensuring a representative sample.

\longonly{One condition on the clusters complements Theorem~\ref{thm:fps} (proof below). Let $\delta_{\mathit{intra}}$ be the largest distance between two documents in the same cluster, i.e., $\delta_{\mathit{intra}} = \max_{1 \le j \le k}\, \max_{d_a, d_b \in C_j} \mathit{dis}_{a,b}$, and let $\delta_{\mathit{inter}} = \min_{1 \le j \le k} \delta_j$ be the smallest separation; the clusters are {\em well-separated} if $\delta_{\mathit{inter}} > 2\,\delta_{\mathit{intra}}$. For well-separated clusters, the sample contains at most one document from each cluster; together with Theorem~\ref{thm:fps}, the sample then contains exactly one document from each cluster whose separation exceeds $g^{*}$.}
Theorem~\ref{thm:fps} does not bound the sample size: Algorithm~\ref{alg:fps} may admit several documents from the same cluster. This is permissible since it does not leave any cluster unrepresented. We therefore cap the sample size at $\min(|S|,\, 0.1\,|\mathcal{D}|)$, where $|S|$ is the sample size Algorithm~\ref{alg:fps} returns. The cap keeps the cost of rule generation low and is empirically effective, as observed in Section~\ref{sec:experiments}.

\longonly{

We now prove Theorems~\ref{thm:fps} and~\ref{thm:fps-group}. We work in the defined distance space, where each document $d_i$ is the point $V_i$ and distances are the cosine distance $\mathit{dis}_{i,j}$. Write $p_1, \dots, p_T$ for the documents Algorithm~\ref{alg:fps} admits, in order ($p_1$ the seed), and $g_t$ ($2 \le t \le T$) for the gap of $p_t$ at its admission (Line~\ref{ln:fps-gap}); when Algorithm~\ref{alg:fps} stops at Line~\ref{ln:fps-stop}, the condition there gives $g^{*} < 0.5\, g_T$; if the condition never holds, the loop ends with $S = \mathcal{D}$ and both theorems hold trivially. The argument never invokes the triangle inequality, which the cosine distance does not satisfy.

\emph{Step 1: gaps are non-increasing, $g_2 \ge g_3 \ge \dots \ge g_T \ge g^{*}$.} Admitting a document can only reduce the distance of every remaining document to $S$, and the document attaining the previous maximum is no longer a candidate, so the next maximum does not grow.

\emph{Step 2: diversity.} Consider two picks $p_a, p_b$ with $a < b$. When $p_b$ was admitted, $p_a$ was already in $S$, so $\mathit{dis}_{p_a, p_b} \ge g_b \ge g_T > 2\,g^{*}$, where the last inequality is the condition in Line~\ref{ln:fps-stop}.

\emph{Step 3: coverage.} By Lines~\ref{ln:fps-pick}--\ref{ln:fps-gap}, $g^{*}$ is the largest distance from a document outside $S$ to the closest document in $S$, so every document outside $S$ is within $g^{*}$ of some document in $S$. Now suppose a union $U$ of clusters with separation greater than $g^{*}$ had no document in $S$; then every document in $U$ would be at distance greater than $g^{*}$ from every document in $S$, all of which lie outside $U$, contradicting Step 3. Theorem~\ref{thm:fps} is the case where $U$ is a single cluster, and Theorem~\ref{thm:fps-group} the case where $U$ is the union of two. The two properties mirror Gonzalez's analysis of farthest-point traversal for $k$-center~\cite{gonzalez1985clustering}, where the selected points form both a packing and a cover; the condition in Line~\ref{ln:fps-stop} certifies the factor of two between the two radii.


\emph{At most one document per cluster.} Suppose the clusters are well-separated. First, while some cluster has no document in $S$, every pick lands in such a cluster: each document of a cluster that already has a document in $S$ is within $\delta_{\mathit{intra}}$ of that document, while every document of a cluster with none is at distance at least $\delta_{\mathit{inter}} > \delta_{\mathit{intra}}$ from every document in $S$, so farthest-point selection (Line~\ref{ln:fps-pick}) prefers the latter. Second, once every cluster has a document in $S$, the next candidate is within $\delta_{\mathit{intra}}$ of its cluster's document in $S$, while $g_{\mathrm{prev}}$ is at least $\delta_{\mathit{inter}}$, since the last admitted document was picked while its cluster had no document in $S$; well-separation gives $\delta_{\mathit{intra}} < 0.5\,\delta_{\mathit{inter}} \le 0.5\,g_{\mathrm{prev}}$, so the condition in Line~\ref{ln:fps-stop} holds and Algorithm~\ref{alg:fps} stops. The sample therefore never contains two documents from the same cluster.

\emph{Remark (uneven separations).} A cluster whose separation is at most $g^{*}$ can be missed, e.g., when one cluster sits far from the rest and another lies close to a sampled one; this is the identifiability limit of distance-based clustering. Enriching the distance features (stronger embeddings or finer chunking) widens the separations and mitigates it. Finally, for $k = 1$ the condition in Line~\ref{ln:fps-stop} requires $|S| \ge 2$, so \sys draws one extra document, a harmless over-sample.
}



\subsection{Rule Application Strategy}
\label{sec:rule-application}

While \sys generates a sample to reduce overfitting, the selected rules may still fail on unsampled documents. We now present a rule application strategy to further improve \sys's robustness. It uses a cheap proxy to detect when a refined rule is unlikely to reproduce the oracle answer, and then falls back to the rule set returned by the agent during rule generation. It also picks a subset of the refined rules to further reduce cost.

Recall that rule generation produces a rule set $R$, and rule refinement (Algorithm~\ref{alg:refine}) selects a subset $R' \subseteq R$, where $R' = \{r_1,\dots,r_k\}$. Algorithm~\ref{alg:refine} adds these rules to $R'$ one at a time, each time choosing the rule of highest {\em cost-effectiveness}: the number of sampled documents it is correct on that no rule already in $R'$ is correct on, divided by its cost ratio $\mathit{cr}(r, D_s)$ (Line~\ref{ln:pick}). We thus treat $R'$ as the ordered list $[r_1,\dots,r_k]$ in that order. We now present a cascade rule application with a fallback.

Given an unsampled document $d_i$, \sys scans $R'$ in the increasing order of the index. After scanning the first $j$ rules, it applies them to $d_i$ and forms the union of their returned spans, $U_j = \bigcup_{l=1}^{j} d_i^{l}$. It then asks a cheap proxy $P$ (e.g., GPT-5.4 mini) whether $U_j$ can answer $Q$, using the prompt in Figure~\ref{fig:proxy-prompt}. If $P$ returns true, \sys stops and calls the oracle $O$ on $U_j$ to produce the answer $A_i$; otherwise it adds the next rule and repeats. If no prefix passes the proxy after all $k$ rules have been tried, \sys falls back to the rule set $R$ and runs the oracle on the union $\bigcup_{r_l \in R} d_i^{l}$ of all its spans. 

Falling back to $R$ is reasonable because $R$ is accurate (matching the accuracy of the golden baseline in Table~\ref{tab:rulegen-results}) at a modest cost (a cost ratio of $0.07$ to $0.39$). 
Using a cheap model as the proxy $P$ is effective. First, the proxy is far cheaper than the oracle, so verifying at each step adds little cost. Second, the task is simple: the refined rules in $R'$ are cost-effective, so the span they return is short, and deciding whether a short, focused span answers $Q$ is far simpler than answering $Q$ itself. 

\begin{figure}
\centering
\begin{tcolorbox}[colback=white,colframe=black!60,boxrule=0.5pt,arc=0pt,left=5pt,right=5pt,top=4pt,bottom=4pt]
\footnotesize\ttfamily
\textbf{\sys-Proxy-Verification-Prompt:} You are given a question \ph{Question} and a span of text \ph{Text} retrieved from a document. Decide whether \ph{Text} contains enough information to answer \ph{Question}.

\textbf{Output:} Return \texttt{true} if the text is sufficient to answer the question, and \texttt{false} otherwise.
\end{tcolorbox}
\vspace{-5mm}
\caption{\small The proxy-verification prompt. \ph{Question} and \ph{Text} are placeholders for the task $Q$ and the union of retrieved spans.}
\label{fig:proxy-prompt}
\vspace{-2mm}
\end{figure}

\shortonly{\vspace{-2mm}}
\section{Experiments}
\label{sec:experiments}

We evaluate \sys on six real-world datasets, describing the setup in Section~\ref{sec:exp-setup} and presenting the results in Section~\ref{sec:exp-results}.

\shortonly{\vspace{-2mm}}
\subsection{Experimental Setup}
\label{sec:exp-setup}

\subsubsection{Datasets}
\label{sec:exp-datasets}

We evaluate \sys on six datasets from distinct real-world domains (Table~\ref{tab:dataset-stats}): \emph{Court}~\cite{courtlistener} (U.S.\ federal appeal opinions), \emph{FinanceBench}~\cite{islam2023financebench} (SEC 10-K/10-Q filings), \emph{NoPV}~\cite{phmsanopv} (PHMSA notices of probable violation), \emph{OfficeQA}~\cite{opsahl2026officeqa} (U.S.\ Treasury Bulletins), \emph{Product}~\cite{emaepar} (EMA drug product information), and \emph{Tropic}~\cite{nhctcr} (NHC tropical cyclone reports). Queries target each domain's structured values, such as case metadata and dispositions, financial figures, cited CFR sections, fiscal indicators, drug attributes, and storm intensity and casualties. The datasets vary widely in scale: documents average from roughly $3{,}300$ to $110{,}000$ tokens, and we include $12$ to $16$ queries per collection. Although each collection is relatively small, \sys already yields substantial savings as shown later. Its cost savings  grow with collection size, and we further extrapolate its cost to $1{,}000$-document collections (Section~\ref{sec:exp-results}).

The query workload spans five types: single-hop lookup (e.g., one stated value), multi-hop lookup (e.g., several named facts), enumeration (e.g., all members of a set), aggregation (e.g., a set reduced to a scalar), and classification (e.g., a label from a closed set). Single-hop lookups form $35\%$ (Court) to $76\%$ (Product) of the queries, $56\%$ on average; the other types often require reasoning over multiple spans across the document. Here, we consider a single query over a collection, while jointly handling multiple queries is left to future work, where semantically related queries may share rules and open further optimization opportunities.

For FinanceBench and OfficeQA, we use the ground truth from  existing benchmarks~\cite{islam2023financebench,opsahl2026officeqa}. For the remaining datasets, we label the answers with two independent LLM labelers, an agent and a vanilla LLM call, both powered by top-of-the-line LLM Claude Opus 4.7, given the complete document as input. We manually verify the answers on which the two labelers disagree. 



\begin{figure*}
\centering
\includegraphics[width=\textwidth]{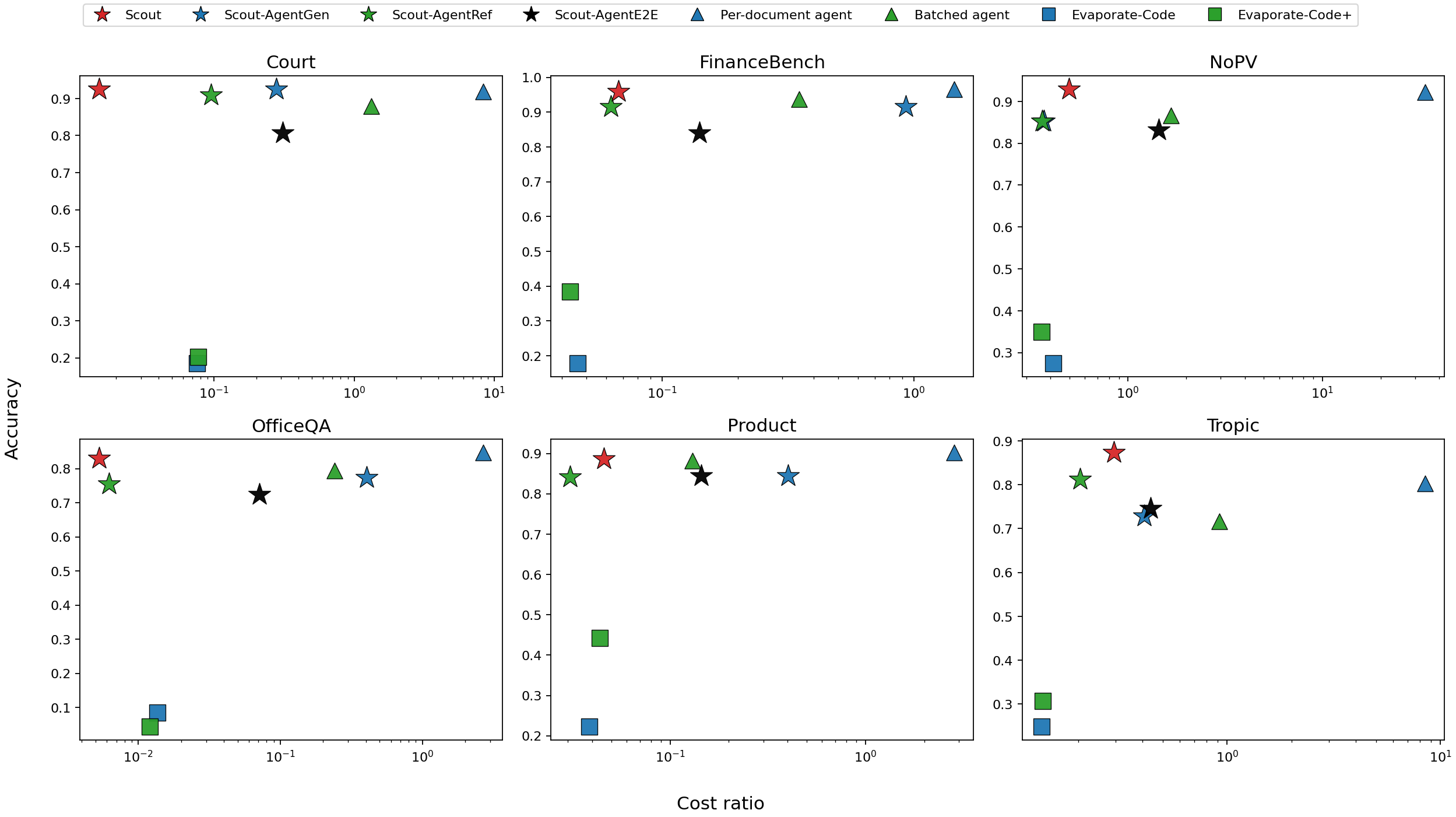}
\vspace{-7mm}
\caption{\small Accuracy versus amortized cost ratio across the six datasets. Up and to the left is better (higher accuracy, lower cost). }
\label{fig:overall}
\vspace{-3mm}
\end{figure*}

\begin{table}[t]
\centering
\small
\setlength{\tabcolsep}{4pt}
\renewcommand{\arraystretch}{1.15}
\begin{tabular*}{\columnwidth}{@{\extracolsep{\fill}}lrrr@{}}
\toprule
\textbf{Dataset} & \textbf{\# Docs} & \textbf{Avg tokens} & \textbf{\# Queries} \\
\midrule
Court        & 294 & 11{,}760  & 13 \\
FinanceBench & 100 & 60{,}539  & 12 \\
NoPV         & 242 & 3{,}293   & 12 \\
OfficeQA     & 200 & 110{,}604 & 16 \\
Product      & 200 & 34{,}970  & 13 \\
Tropic       & 200 & 9{,}789   & 14 \\
\bottomrule
\end{tabular*}
\caption{\small Datasets used in the experiments. \emph{\# Docs} is the number of documents, \emph{Avg tokens} the average number of tokens per document, and \emph{\# Queries} the number of extraction queries per dataset.}
\label{tab:dataset-stats}
\vspace{-4mm}
\end{table}

\shortonly{\vspace{-2mm}}
\subsubsection{Baselines and Ablations}
\label{sec:exp-baselines}

We evaluate \sys against two groups of methods: baselines and ablations that replace \sys's components with an agent.

\noindent{\bf Baselines.} For document extraction, agentic question answering is the state of the art, outperforming retrieval-augmented (RAG) and long-context baselines by up to $16$ points on long-document QA benchmarks~\cite{arag,longagent,graphreader,sun2025docagent}. 
We instantiate the agent as Codex~\cite{codex}, powered by GPT-5.4, and give it all the tools Codex provides. A few tools retrieve spans by keyword and regex matching or by embedding similarity, and reason over them. The agent then answers a query iteratively: at each step it issues a tool call, observes the result, and chooses the next call, until it has gathered enough evidence to answer. Note that RAG can be viewed as a special case of the agent, retrieving once based on embedding and then answering.

Using this step, we adopt varied agent-based methods as our baselines. The \emph{per-document agent} runs this agent on one query and one full document at a time, while the \emph{batched agent} instead presents one query together with all documents at once and lets the agent answer the query per document. Finally, \emph{Evaporate}~\cite{evaporate} uses an LLM to synthesize Python programs for data extraction directly, invoking no LLM within the programs; we consider two variants: \emph{Evaporate-Code} synthesizes a single program based on sampled documents, while \emph{Evaporate-Code+} synthesizes many and ensembles their extractions with weak supervision. 

\noindent{\bf Ablations.} \sys runs a four-step pipeline: sampling, rule generation, rule refinement, and rule application. To isolate the contribution of each of \sys's components, we compare it against ablations that hand some steps to an agent while keeping the rest of the pipeline fixed (Table~\ref{tab:ablation-strategies}). Rule application is always performed by \sys, and each ablation delegates a different part of the pipeline to the agent. \emph{\sys-AgentRef} delegates rule refinement to the agent, under the same accuracy and cost objective as Problem~\ref{prob:refine} (i.e., to minimize the cost of rules while enforcing that their accuracy matches that of the rule set returned by \sys). \emph{\sys-AgentGen} delegates both rule generation and refinement: the agent instead learns rules directly from the sampled documents, under that same objective.    \emph{\sys-AgentE2E} delegates the entire offline pipeline (steps~1 to~3), where the agent learns the rules end-to-end from all the data and effectively samples on its own. Comparing these variants reveals whether \sys's specialized components outperform a general agent performing the same steps. 
For the agents used in these ablations, we provide tools that compute the accuracy and cost of any rule set, as defined in Section~\ref{sec:problem-rules}, to let the agent evaluate a rule set.\shortonly{ We defer the detailed prompts of ablations to our technical report~\cite{scoutextended}.}

\longonly{
We describe each ablation and give its prompt below. Each replaces one or more of \sys's engineered stages with a general-purpose coding agent (Codex, backed by GPT-5.4) that is given the repository tools and left to solve the same sub-problem on its own. All three are handed the same accuracy and cost-ratio definitions as \sys (Section~\ref{sec:problem-rules}), exposed as callable tools; in the prompts below these tool interfaces appear as placeholders (e.g., \ph{Accuracy specification}, \ph{Cost ratio specification}).

\noindent{\bf \sys-AgentGen} delegates both rule generation and refinement. The agent is prompted as a document rule engineer and given the task \ph{Question}, the enriched sampled documents \ph{Documents}, and their oracle answers \ph{Answers}; it uses the repository tools to load documents, measure accuracy and cost, and iterate toward a small, cost-optimized rule set, where \ph{Accuracy specification} and \ph{Cost ratio specification} are the metric tools of Section~\ref{sec:problem-rules}.

\begin{tcolorbox}[colback=white,colframe=black!60,boxrule=0.5pt,arc=0pt,left=5pt,right=5pt,top=4pt,bottom=4pt,breakable]
\footnotesize\ttfamily
\textbf{\sys-AgentGen-Prompt:} You are a document rule engineer working in the project repository. Given a question \ph{Question}, the enriched sampled documents \ph{Documents}, and their oracle answers \ph{Answers}, generate and refine a small set of Python span-retrieval rules \texttt{def rule\_name(doc)} that return the answer-bearing spans of any unseen document from the same collection. Use every tool available (read and write files, execute Python, run shell commands, and search) to load documents, measure cost and accuracy, and iterate; do not reason in isolation.

\textbf{Objectives (in priority order):} (1)~reach merge accuracy at least $0.95$ on the sampled documents; (2)~minimize the average cost ratio (target below $0.1$), tightening high-cost rules once accuracy is met, with accuracy always winning.

\textbf{Metrics (tools):} \ph{Accuracy specification} and \ph{Cost ratio specification}, where accuracy runs the oracle on the union of the returned spans and checks it against the answer, and cost ratio is the fraction of the document's tokens retrieved. A fast substring proxy may be used to iterate cheaply before the final oracle check.

\textbf{Workflow:} study where the answer sits in each sampled document; write the broadest rule first and test its coverage and cost; diagnose uncovered documents and add a targeted rule only when it covers at least two of them; verify merge accuracy with the oracle and repeat until the target is met; finally tighten any high-cost rule without losing accuracy.

\textbf{Requirements:} each rule is self-contained (imports inside the body), must never raise, and returns an empty list when nothing matches; prefer a few extra spans over missing the answer. Locate the answer using physical location (page), semantic location (section header or path), keyword proximity, table cells, typography, structural position, and any other recurring pattern.
\end{tcolorbox}

\noindent{\bf \sys-AgentRef} delegates rule refinement. The agent is given the task \ph{Question} and the pre-generated rule pool \ph{Rules}, and selects a small subset with the repository tools; \ph{Budget} bounds the paid \texttt{verify\_accuracy} calls, and \ph{Accuracy specification}, \ph{Cost ratio specification}, and \ph{Coverage specification} are the metric tool interfaces of Section~\ref{sec:problem-rules}.

\begin{tcolorbox}[colback=white,colframe=black!60,boxrule=0.5pt,arc=0pt,left=5pt,right=5pt,top=4pt,bottom=4pt,breakable]
\footnotesize\ttfamily
\textbf{\sys-AgentRef-Prompt:} You are given a question \ph{Question} and a pre-generated rule pool \ph{Rules}. Select a small subset of rules that preserves the merge accuracy of the full pool on the sampled documents while keeping total cost low and per-rule coverage high. Invoke the repository tools below (one per call); you may write small helper scripts.

\textbf{Hard constraint:} the merge accuracy of your subset must equal that of the full pool on every sampled document the full pool answers correctly.

\textbf{Soft targets (traded against each other):} minimize the summed cost ratio of the subset; maximize the minimum per-rule coverage (broad rules generalize better); keep the subset small.

\textbf{Tools:} \texttt{list\_rules} (pool with one-line descriptions); \texttt{compute\_cost} \ph{Cost ratio specification}; \texttt{compute\_coverage} \ph{Coverage specification}; \texttt{inspect\_rule}; and the paid \texttt{verify\_accuracy} \ph{Accuracy specification}, limited to \ph{Budget} calls. The free tools use no oracle; \texttt{verify\_accuracy} runs the oracle on the sampled documents.

\textbf{Loop:} snapshot the pool with cost and coverage; propose an initial subset by coverage-per-cost, preferring rules with broad descriptions over layout-specific ones; verify accuracy; for any missed document, inspect and add the most promising covering rule, then re-verify; once the hard constraint holds, drop-test the most expensive rule and remove it if accuracy is preserved; stop when the soft targets are reasonable or the budget is exhausted.
\end{tcolorbox}

\noindent{\bf \sys-AgentE2E} delegates the entire offline pipeline (sampling, generation, and refinement) over the raw corpus. The agent is given the task \ph{Question} and the plain-text corpus \ph{Documents}, and samples, generates, and refines rules end to end; \ph{Budget} bounds the paid \texttt{verify-accuracy} calls, and \ph{Accuracy specification} and \ph{Cost ratio specification} are the metric tool interfaces of Section~\ref{sec:problem-rules}.

\begin{tcolorbox}[colback=white,colframe=black!60,boxrule=0.5pt,arc=0pt,left=5pt,right=5pt,top=4pt,bottom=4pt,breakable]
\footnotesize\ttfamily
\textbf{\sys-AgentE2E-Prompt:} You are given a question \ph{Question} and a corpus of plain-text documents \ph{Documents}. Working end to end, sample documents, generate Python retrieval rules \texttt{def rule\_name(doc)}, and refine them into a small rule set that answers the question at low cost. Use the repository tools below and keep all writes inside the rules directory.

\textbf{Tools:} \texttt{list-docs} and \texttt{read-doc} (inspect a seed sample); \texttt{compute-cost} \ph{Cost ratio specification}; \texttt{inspect-rule}; and the paid \texttt{verify-accuracy} \ph{Accuracy specification}, limited to \ph{Budget} calls. Documents are plain text, exposed as pages, paragraphs, and lines.

\textbf{Rule contract:} one file per rule, each defining a single \texttt{def rule\_name(doc)} that returns a list of span dicts (each with a \texttt{text} field) and wraps its body in \texttt{try/except} returning an empty list.

\textbf{Workflow:} inspect a small seed sample; write one to three broad rules; measure cost and verify accuracy on the working sample; expand the sample only to diagnose failures; stop when the rules are stable or further tightening hurts coverage.
\end{tcolorbox}
}


\begin{table}[t]
\centering
\footnotesize
\setlength{\tabcolsep}{4pt}
\renewcommand{\arraystretch}{1.15}
\begin{tabularx}{\columnwidth}{|l|C|C|C|C|}
\hline
\textbf{Strategy} & \textbf{Sample} & \textbf{Generate} & \textbf{Refine} & \textbf{Apply} \\
\hline
\sys          & \sys  & \sys  & \sys  & \sys \\
\hline
\sys-AgentRef & \sys  & \sys  & Agent & \sys \\
\hline
\sys-AgentGen & \sys  & \multicolumn{2}{c|}{Agent} & \sys \\
\hline
\sys-AgentE2E & \multicolumn{3}{c|}{Agent} & \sys \\
\hline
\end{tabularx}
\caption{\small \sys and its agentic ablations}
\label{tab:ablation-strategies}
\vspace{-5mm}
\end{table}


\shortonly{\vspace{-1mm}}
\subsubsection{Metrics}
\label{sec:exp-metrics}
We report three metrics per dataset: \emph{accuracy}, end-to-end \emph{latency}, and \emph{cost} of each strategy. For an extraction query $Q$ over a document collection $\mathcal{D}$, we measure each metric per document and average it over the documents, then over all extraction queries. Some strategies run in two stages: an offline stage that learns the rules, with optional sampling and rule refinement (as in \sys and its ablations), and an online stage that applies them. 
Let $\mathit{cost}_{\mathrm{off}}(Q)$ and $\mathit{cost}_{\mathrm{on}}(Q, d_j)$ be the offline cost for $Q$ and the cost of applying the rules to document $d_j$. The reported average cost amortizes the one-time offline cost across all $n=|\mathcal{D}|$ documents: $\mathit{cost}(Q) = \frac{1}{n}\sum_{d_j \in \mathcal{D}} \mathit{cost}_{\mathrm{on}}(Q, d_j) + \frac{\mathit{cost}_{\mathrm{off}}(Q)}{n}.$
The average latency $\mathit{lat}(Q)$ is defined analogously, with $\mathit{lat}_{\mathrm{off}}$ and $\mathit{lat}_{\mathrm{on}}$ in place of the costs. Baselines have no offline stage, so their second term is zero. By default we set $\alpha = 0.05$ (Problem~\ref{prob:refine}), so the refined rules' accuracy falls at most $5\%$ below the full pool's.

Given $Q$, when measuring accuracy, we measure whether the span $d_j^i$ returned by a rule $r_i$ on a document $d_j$ can reproduce the oracle answer (i.e., the ground truth), i.e., $O(d_j^i,Q) = O(d_j,Q)$. Here, we use an LLM-as-a-judge~\cite{gu2026survey,zheng2023llmjudge} to decide whether $O(d_j^i,Q)$ is equivalent to the ground truth, allowing small syntactic variations, and GPT-5.4 is used as the judge. In particular, let $I(O(d_j^i,Q), O(d_j,Q))$ be an indicator function that returns True if $O(d_j^i,Q)$ and $O(d_j,Q)$ are lexically identical or semantically equivalent. This approach effectively handles equivalent but non-identical responses. 

\subsection{Experimental Results}
\label{sec:exp-results}

\noindent{\bf Experiment 1: \sys vs.\ baselines.} Figure~\ref{fig:overall} reports the accuracy ($y$-axis) and the amortized cost ratio $\mathit{cost}(Q)$  ($x$-axis, log scale), one panel per dataset. A point in the top-left corner is ideal: high accuracy at low cost. 

\noindent{\bf Agent baselines.} \sys {\bf \em matches the accuracy} of the {\em per-document agent}, the strongest but most expensive baseline, while spending far less. Averaged over the six datasets, \sys reaches $0.9$ accuracy against $0.89$ for the per-document agent, leading on Court, NoPV, and Tropic, and trailing by at most $1.6$ points elsewhere, yet its amortized cost is {\bf \em one to two orders of magnitude lower} ($22\times$ to $552\times$). This is because the per-document agent reads each full document and issues many tool calls in multiple iterations, so its cost ratio exceeds $1$, consuming more than the whole document, whereas \sys applies cost-effective rules and invokes the oracle only on a small span. The {\em batched agent} is cheaper but only $0.85$ accurate on average, since packing all documents into one context dilutes the model's attention, and it remains $3\times$ to $88\times$ more expensive than \sys.

\noindent{\bf Code-rule baseline.} Against {\em Evaporate-Code+}, the strongest version of Evaporate variant, \sys is {\bf\em\boldmath $61$ points more accurate} ($0.90$ vs.\ $0.29$) at a comparable amortized cost. This is because Evaporate instructs LLMs to generate programs that emit answers directly with no final oracle pass, so it succeeds only when the answer is a substring of the document; queries that otherwise require synthesizing from the context fail, dropping its accuracy below $0.3$ on every dataset. 

\begin{figure*}
\centering
\includegraphics[width=\textwidth]{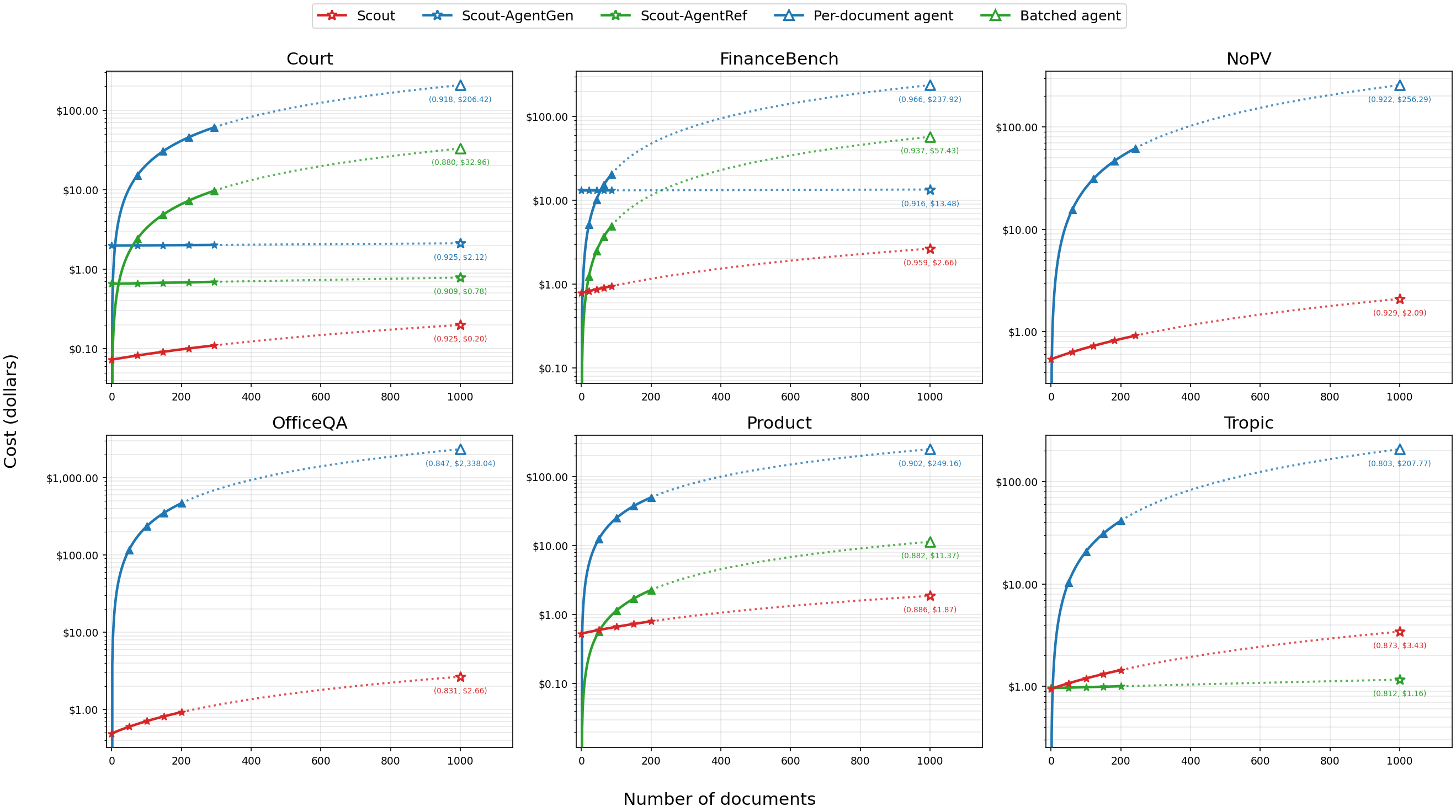}
\vspace{-8mm}
\caption{\small Total cost as the document collection grows. {\em Only strategies within $5\%$ of the per-document agent's accuracy are shown.}}
\vspace{-3mm}
\label{fig:scalability}
\end{figure*}

\noindent{\bf Experiment 2: \sys vs.\ its ablations.} We now ask whether each of \sys's own components is necessary. 


\noindent{\bf Rule refinement.} \sys-AgentRef differs from \sys only in refinement, so comparing the two isolates the contribution of \sys's refinement algorithm. At a comparable cost ratio, \sys is {\bf\em more accurate on all six datasets} (Figure~\ref{fig:overall}), by $5$ points on average. When \sys-AgentRef reaches a slightly lower cost ratio, it pays with a large accuracy drop, as on Tropic, while on Court \sys is both more accurate and cheaper. The reason is that rule refinement is NP-hard (Theorem~\ref{thm:refine-hardness}): the agent searches a large space of rule subsets heuristically without guarantees, whereas \sys solves it with a greedy algorithm that is provably near-optimal.
\noindent{\bf Rule generation.} \sys is {\bf\em more accurate and cheaper} than \sys-AgentGen. On average \sys is $6$ points more accurate, and the gap widens to $14$ points on the most heterogeneous dataset, Tropic. 
This is because \sys generates a large rule set that prioritizes accuracy over cost in the rule generation step, and then reduces cost while preserving accuracy through rule refinement. When the refined rules don't contain the answers, \sys falls back to the initial rule set, and thus generalizes well to unseen documents. In contrast, \sys-AgentGen directly instructs an agent to generate an optimized rule set with high accuracy and low cost based on the samples, which may overfit.

In terms of the cost of rule generation, \sys is cheaper because its agent generates this rule set in a few iterations, whereas \sys-AgentGen solves a harder optimization problem, maximizing accuracy while constraining the cost of the generated rules, and thus incurs significantly more iterations, increasing the cost of both program synthesis and verification. Even so, the rules returned by \sys are more effective: at matched accuracy (e.g., $0.925$ on Court), \sys's cost ratio is about an order of magnitude lower.

\begin{table*}[t]
\centering
\small
\setlength{\tabcolsep}{4pt}
\renewcommand{\arraystretch}{1.25}
\resizebox{\textwidth}{!}{%
\begin{tabular}{l*{6}{rrr}}
\toprule
 & \multicolumn{3}{c}{\textbf{Court}} & \multicolumn{3}{c}{\textbf{FinanceBench}} & \multicolumn{3}{c}{\textbf{NoPV}} & \multicolumn{3}{c}{\textbf{OfficeQA}} & \multicolumn{3}{c}{\textbf{Product}} & \multicolumn{3}{c}{\textbf{Tropic}} \\
\cmidrule(lr){2-4}\cmidrule(lr){5-7}\cmidrule(lr){8-10}\cmidrule(lr){11-13}\cmidrule(lr){14-16}\cmidrule(lr){17-19}
\textbf{Strategy} & sAcc & uAcc & $\delta$ & sAcc & uAcc & $\delta$ & sAcc & uAcc & $\delta$ & sAcc & uAcc & $\delta$ & sAcc & uAcc & $\delta$ & sAcc & uAcc & $\delta$ \\
\midrule
\sys          & 0.965 & 0.922 & $+0.043$ & 0.960 & 0.959 & $+0.001$ & 0.939 & 0.928 & $+0.011$ & 0.880 & 0.826 & $+0.054$ & 0.877 & 0.886 & $-0.009$ & 0.877 & 0.873 & $+0.004$ \\
\sys-AgentRef & 0.973 & 0.904 & $+0.069$ & 0.950 & 0.904 & $+0.046$ & 0.942 & 0.845 & $+0.097$ & 0.840 & 0.746 & $+0.094$ & 0.892 & 0.836 & $+0.057$ & 0.907 & 0.801 & $+0.106$ \\
\sys-AgentGen & 0.965 & 0.922 & $+0.044$ & 0.960 & 0.903 & $+0.057$ & 0.942 & 0.845 & $+0.097$ & 0.830 & 0.773 & $+0.057$ & 0.881 & 0.840 & $+0.041$ & 0.761 & 0.725 & $+0.035$ \\
\bottomrule
\end{tabular}%
}
\caption{\small Overfitting of the \sys-family strategies. All values are averaged over each dataset's queries. }
\label{tab:overfit}
\vspace{-5mm}
\end{table*}

\noindent{\bf Sampling.} \sys-AgentE2E delegates sampling, generation, and refinement, and is the weakest variant: it is both less accurate and more expensive than \sys on all six datasets, e.g., on Court it scores $0.81$ at cost ratio $0.31$ against $0.93$ at $0.015$ for \sys, because a general agent fails to sample representatively. Its consistent shortfall confirms that \sys's offline components, and its representative sampling in particular, are effective. 

\noindent{\bf Experiment 3: Scalability.} Figure~\ref{fig:scalability} reports the total cost in dollars (log-scale $y$-axis) of answering a query as the document collection grows, from the measured collection size up to $1{,}000$ documents (solid lines measured, dotted lines extrapolated), for \sys, its ablations, and the baselines. To make the cost comparison fair, we include only strategies within $5\%$ of the  strategy with the best accuracy, i.e., the per-document agent. As the collection grows, \sys's total cost stays {\bf \em nearly flat} while the baselines grow linearly.

\noindent{\bf Cost at scale.} At $1{,}000$ documents, \sys answers the query for a few dollars on every dataset (\$$0.20$ to \$$3.43$), whereas the per-document agent costs hundreds to thousands of dollars (\$$206$ to \$$2{,}338$), making {\bf\em\boldmath \sys\ $61\times$ to over $1000\times$ cheaper} at that scale, e.g., \$$0.20$ vs.\ \$$206$ on Court and \$$2.66$ vs.\ \$$2{,}338$ on OfficeQA. The batched agent is cheaper than the per-document agent but its cost still grows linearly, reaching \$$11$ to \$$57$ at $1{,}000$ documents, one to two orders of magnitude above \sys. This is because \sys spends only a one-time offline cost and then a tiny online cost per document, thanks to the small spans returned by the rules. The agent baselines instead read each full document (possibly in multiple passes), so their total cost rises linearly and quickly dominates \sys's fixed offline cost. Moreover, among the strategies that are cheap at scale, \sys is also the most accurate: on Tropic, \sys and \sys-AgentRef both cost only a few dollars at $1{,}000$ documents (\$$3.43$ and \$$1.16$), yet \sys is  $6$ points more accurate. 

\noindent{\bf Cost breakdown of \sys.} \sys's cost splits into a fixed one-time offline cost to learn the rules and a small  per-document online cost to apply them. 
Applying the rules to one document is cheap, only about \$$0.0016$ on average. At the measured collection sizes the offline cost is $53\%$ to $83\%$ of the total, but since it is fixed, its fraction shrinks as the collection grows: at $1{,}000$ documents it falls to $18\%$ to $37\%$, and the online cost dominates.

\noindent{\bf Experiment 4: Overfitting analysis.} \sys learns rules from a small sample, so we test whether they overfit it. Table~\ref{tab:overfit} reports the accuracy of  \sys and  its ablations on the sampled (\emph{sAcc}) and held-out (\emph{uAcc}) documents, and their gap $\delta = \text{sAcc} - \text{uAcc}$. A  larger positive $\delta$ means worse generalization. We omit \sys-AgentE2E, which samples, generates, and refines rules in a single agent call: it does not expose which documents it sampled, and forcing it to do so explicitly in the prompt might change the method. 

\sys overfits the least: its average gap is $0.017$, against $0.055$ for \sys-AgentGen and $0.078$ for \sys-AgentRef, the smallest on every dataset, while its ablations have gaps up to $0.097$ and $0.106$. \sys-AgentGen tends to overfit as it instructs agents to generate optimized rules (with high accuracy and low cost) based on the sample, whereas \sys generates a large pool of rules prioritizing accuracy, which is used as the fallback when the refined rules don't contain answers. Rules returned by \sys additionally  prioritize high coverage in the rule refinement, which alleviates overfitting.


\begin{table}[t]
\centering
\footnotesize
\setlength{\tabcolsep}{4pt}
\renewcommand{\arraystretch}{1.2}
\resizebox{\columnwidth}{!}{%
\begin{tabular}{@{}lrrrrrr@{}}
\toprule
\textbf{Dataset} & \textbf{Embed} & \textbf{Sampling} & \textbf{Rule gen.} & \textbf{Refine} & \textbf{Offline} & \textbf{Online} \\
 & (s/doc) & (s) & (s) & (s) & (s/query) & (s/doc) \\
\midrule
Court        & 0.29 & 85.3 & 40.5 & 188 & 313.8 & 3.78 \\
FinanceBench & 0.35 & 35.0 & 35.9 & 219 & 289.9 & 2.19 \\
NoPV         & 0.22 & 53.2 & 49.9 & 192 & 295.1 & 2.13 \\
Product      & 0.39 & 78.0 & 38.1 & 134 & 250.1 & 2.12 \\
Tropic       & 0.26 & 52.0 & 52.8 & 244 & 348.8 & 2.30 \\
\bottomrule
\end{tabular}%
}
\caption{\small Latency breakdown of \sys. \emph{Embed} is the per-document embedding time within sampling. All times are in seconds.}
\label{tab:latency}
\vspace{-4mm}
\end{table}

\noindent{\bf Experiment 5: Latency breakdown of \sys.} Table~\ref{tab:latency} reports the latency breakdown of \sys. The three offline steps (sampling, rule generation, and refinement) run once per query, and \emph{Offline} is their sum. \emph{Online} is the time to apply the rules to a single document.

\sys's offline latency is a few hundred seconds per query. This time is spent only once and reused across the collection. The online latency is about $2$ to $4$ seconds per document. This is several times faster than the per-document agent, which takes $8$ to $19$ seconds per document. Embedding the documents during sampling is the other one-time cost: at $0.2$ to $0.4$ seconds per document, it adds up to $35$ to $85$ seconds. 

\section{Related Work}
\label{sec:related}


We review work relevant to \sys: document extraction, LLM-powered data systems, and cost-optimized LLM data processing.

\noindent{\bf Document extraction.}
A rich line of work extracts structured values from documents; prior methods differ in what they assume or require.
Many assume documents share a rigid structure: TWIX~\cite{lin2025visual} recovers tables from a tabular template, ZenDB~\cite{zendb} and SHED~\cite{shed}  query documents that follow a hierarchical template, while Doctopus~\cite{doctopus} extracts tables under a budget, and InstrucTE~\cite{schemaie} maps heterogeneous tables into a target schema. Work on web extraction from a decade ago relies on the HTML markup of template-generated pages~\cite{roadrunner,lixto,exalg,fivatech,webinvariants}. Such approaches fail when documents do not have a consistent template or markup, whereas \sys does not pose restrictions on such rigid document structure or markups and learns where answers recur automatically. 
Program-based extraction is another alternative. Evaporate~\cite{evaporate} assumes each answer is a substring of the input and synthesizes a program to return them directly, and it fails when the answer must be inferred from context. \sys instead uses programs only to locate the answer for an LLM to read, with guarantees on the accuracy and cost of the programs it generates, and is over $61$ points more accurate.
Finally, the rest work on training an LLM for document extraction focuses on accuracy rather than cost, such as AWS Textract~\cite{textract}, Azure Document Intelligence~\cite{azuredocai}, and vision LLMs such as GPT-4 Vision~\cite{gpt4v}. Applying such models on large document collections does not scale. In contrast, machine-learning-based
 extractors~\cite{glean,layoutlm,form2seq,visualspan,parthasarathy2022landmarks} either need human input (e.g., labels) or do not transfer across domains. \sys needs neither: given only a natural-language query, it discovers where answers recur automatically and returns a small document span to the LLM,  extracting at scale with low cost. 

\noindent{\bf LLM-powered data systems.}
A growing line of systems use LLMs to answer semantic queries over unstructured data.
LOTUS~\cite{patel2025semantic} adds LLM-defined semantic operators (e.g., filters, joins, aggregations) to the relational model.
DocETL~\cite{docetl} rewrites document-processing pipelines with an agentic framework for higher accuracy; Palimpzest~\cite{palimpzest} compiles declarative AI analytics into plans trading off cost, runtime, and quality.
QUEST~\cite{quest} reduces extraction cost via index-based retrieval; AOP~\cite{aop} optimizes operator pipelines online for multi-hop queries over data lakes; ThalamusDB~\cite{thalamusdb} answers SQL with natural-language predicates over multi-modal data, such as image, audio, and text, via approximate query processing with error bounds.
Despite their differences, these systems share one step: to evaluate a semantic operator, they call an LLM to synthesize the data it needs from the source, which amounts to extracting structured data on the fly.
\sys is complementary, and can be viewed to be optimized physical operators for such systems, e.g., supporting a map operation over  documents.  


\noindent{\bf Cost-optimized LLM-powered processing.}
Another line of work reduces the cost of LLM-powered data processing.
Model cascades and routers send each input to the cheapest capable model: FrugalGPT~\cite{frugalgpt} chains models and stops once an answer is confident; Hybrid LLM~\cite{hybridllm} routes each query to a small or large model by predicted difficulty; and RouteLLM~\cite{ong2025routellm} learns such routers from preference data.
BARGAIN~\cite{zeighami2025cut} instead runs cheap models with adaptive sampling and statistical estimation to cut cost while bounding the accuracy loss relative to an oracle.
A related line answers AI queries (e.g., a SQL query with AI predicates) with lightweight proxy models over embeddings~\cite{chung2026100x}, but targets classification tasks rather than data extraction. 
These techniques choose cheaper models when accuracy is not sacrificed, but still send the full input data to an LLM, and tend to escalate to the expensive oracle on long inputs where cheap models degrade.
\sys reduces cost orthogonally: its rules exploit data similarity so that an LLM reads only a small, relevant span of each document. 
\section{Conclusion}
\label{sec:conclusion}

We presented \sys, a scalable document-extraction system that generates effective rules to locate answers with accuracy and cost guarantees. \sys samples a few representative documents; an agent then generates a set of rules that together recover almost all answers, and \sys selects a subset that preserves this accuracy at a fraction of the cost, a selection problem we prove to be NP-hard and solve greedily within a logarithmic factor of optimal. To reduce overfitting of rules learned on the sample, \sys applies the refined rules to an unseen document one at a time, checks with a cheap model whether the retrieved text suffices, and falls back to the full rule set when it does not. On real-world datasets, \sys matches the accuracy of a frontier agent while being orders of magnitude cheaper, and far outperforms prior program-based systems. 

\clearpage

\end{document}